\documentclass[12pt]{article}
\usepackage{amsmath}
\usepackage{graphicx}
\usepackage{enumerate}
\usepackage{amssymb}
\usepackage{chngcntr}
\usepackage[table]{xcolor} 
\usepackage{multicol}
\usepackage{titling}
\usepackage{amsthm}
\usepackage{abstract}
\usepackage{authblk}
\usepackage{makecell}
\usepackage{natbib}
\usepackage[colorlinks=false]{hyperref}
\usepackage{float}
\usepackage{natbib}
\usepackage{graphicx}
\usepackage{amsfonts}
\usepackage{tabu}
\usepackage{caption}
\usepackage{booktabs}
\usepackage{multirow}
\usepackage{comment}
\usepackage{natbib}
\usepackage{url} % not crucial - just used below for the URL 

\newcommand{\blind}{1}

\newcommand\dd{\mathop{}\!\mathrm{d}}

\newcommand{\bfeta}{\boldsymbol{\eta}}

\newcommand{\bftheta}{\boldsymbol{\theta}}

\newcommand{\bfp}{\boldsymbol{p}}

\newcommand{\bfX}{\boldsymbol{X}}

\newcommand{\bfB}{\boldsymbol{B}}

\newtheorem{theorem}{Theorem}[section]
\newtheorem{corollary}{Corollary}[section]

\begin{document}

\def\spacingset#1{\renewcommand{\baselinestretch}%
{#1}\small\normalsize} \spacingset{1}

%%%%%%%%%%%%%%%%%%%%%%%%%%%%%%%%%%%%%%%%%%%%%%%%%%%%%%%%%%%%%%%%%%%%%%%%%%%%%%

% Original title: Semiparametric estimation for multivariate Hawkes processes using dependent Dirichlet processes:  An application to order flow data in financial markets

\if1\blind
{
  \title{\bf Semiparametric estimation for multivariate Hawkes processes using dependent Dirichlet processes:  An application to order flow data in financial markets}
  \author{Alex Ziyu Jiang and Abel Rodríguez\hspace{.2cm}\\
    Department of Statistics, University of Washington}
  \maketitle
} \fi

\if0\blind
{
  \bigskip
  \bigskip
  \bigskip
  \begin{center}
    {\LARGE \bf A Flexible Multivariate Hawkes Process Model For Order Flow Data}
  \end{center}
  \medskip
} \fi

\bigskip
\begin{abstract}
The order flow in high-frequency financial markets has been of particular research interest in recent years, as it provides insights into trading and order execution strategies and leads to better understanding of the supply-demand interplay and price formation. In this work, we propose a semiparametric multivariate Hawkes process model that relies on (mixtures of) dependent Dirichlet processes to  analyze order flow data. Such a formulation avoids the kind of strong parametric assumptions about the excitation functions of the Hawkes process that often accompany traditional models and which, as we show, are not justified in the case of order flow data.  It also allows us to borrow information across dimensions, improving estimation of the individual excitation functions.  To fit the model, we develop two algorithms, one using Markov chain Monte Carlo  methods and one using a stochastic variational approximation.  In the context of simulation studies, we show that our model outperforms benchmark methods in terms of lower estimation error for both algorithms.  In the context of real order flow data, we show that our model can capture features of the excitation functions such as non-monotonicity that cannot be accommodated by standard parametric models.
\end{abstract}

\noindent%
{\it Keywords:}  Dependent Dirichlet Processes, Order Book Dynamics, High-Frequency Trading, Markov chain Monte Carlo, Variational Inference
\vfill

\newpage
\spacingset{1.9} % DON'T change the spacing!
\section{Introduction}
\label{sec:intro}

{In the past two decades, the proliferation of electronic trading systems has led to financial instruments such as stocks and futures being more frequently traded in order-driven markets. In these order-driven markets, buy and sell orders are submitted by traders to an electronic platform, which matches them with the best available offers (e.g., see \citealp{gould2013limit}). Due to the nature of electronic trading, the arrival of orders tends to have high-frequency, with orders usually arriving within milliseconds of each other. The order flow in such markets is of particular research interest, as it not only provides insights to high-frequency trading and order execution strategies \citep{alfonsi2010optimal}, but also aids in understanding the interplay of supply and demand and their roles in price formation \citep{cont2011statistical}. %The \textit{limit order book} (LOB) records the buy and sell orders for a given asset that are currently unfilled, along with the timestamp when such orders were placed. 

%As LOB documents the mechanism for matching buyers and sellers in an order-driven market, various approaches have been used to model the order flow in LOB to better understand the market dynamics, 

A number of approaches have been proposed to analyze data from order-driven markets, including queueing systems \citep{muni2015order}, stochastic partial differential equations \citep{cont2021stochastic} and Hawkes processes \citep{rambaldi2017role}. %In this section, we employ our nonparametric model to study the order flow dynamics of the LOB data \citep{huang2011lobster}. 
In this work, we focus on multivariate Hawkes Processes (MHPs, e.g., see \citealp{Hawkes1971} and \citealp{liniger2009multivariate}) models. 
%Hawkes Processes have been widely applied in a wide range of fields, including seismology \citep{Ogata1988}, finance \citep{bacry2015hawkes}, electronic health records \citep{Choi2015EHR,sun2024learning}, disease transmission \citep{schoenberg2023estimating} and social media analysis \citep{Rizoiu2017}. 
MHPs have been widely used to model sequences of events that demonstrate \textit{self-} and \textit{mutually-exciting} behaviors, i.e., patterns in which the likelihood of events increases after the occurrence of others. MHPs can be characterized through their conditional intensity functions, which describe the instantaneous rate of arrivals of new events. The excitation function is a key module of the conditional intensity function, as it controls how past events cause the conditional intensity function to change, and how such intensity decays when no new events are observed in the meantime. Excitation functions are often modeled parametrically (e.g., using exponential or polynomial functions) that assume monotonic excitation decay \citep{Hawkes1971}. However, the inter-event times in many real-world examples tend to have complex patterns, motivating the need for flexible excitation functions (e.g., see \citealp{Markwick2020} and \citealp{Rodriguez2017}).}

Nonparametric models for the intensity function associated with a Hawkes processes have been investigated in the past. For example, \cite{marsan2008extending} proposed a histogram estimator that can be computed via the expectation-maximization (EM) algorithm. The efficacy of their method is further studied in \cite{sornette2009limits}.  In related work, \cite{lewis2011nonparametric} and
\cite{zhou2013learning} developed  penalized maximum likelihood where the regularization term favors smooth estimates of the intensity functions. Alternatively, \cite{reynaud2010adaptive} developed penalized projection estimators. Their work was extended by \cite{hansen2015lasso}, who  provided non-asymptotic theoretical guarantees on model selection. On yet another track, \cite{bacry2015hawkes} and \cite{bacry2016first} proposed a nonparametric estimation framework for MHP models based on the discretization of a Wiener-Hopf system that relates the excitation functions to the second order statistics of the MHP. Under a Bayesian framework, \cite{zhang2020variational} and \cite{zhou2020efficient} considered flexible models for the excitation functions through Gaussian process priors coupled with squared or sigmoid link functions. These methods lead to flexible models for the excitation function, but their theoretical properties are not well understood. In contrast, \cite{donnet2020nonparametric} modeled the excitation functions using nonparametric mixtures and established posterior concentration rates for Bayesian nonparametric estimation of MHP, under the finite support assumption for the excitation functions.  A related approach was developed in parallel in \cite{Markwick2020}.

In this paper, we introduce a Bayesian semiparametric model for MHPs that builds on the ideas of \cite{donnet2020nonparametric} and \cite{Markwick2020} but addresses various practical questions that have so far remained open.  One challenge associated with the estimation of MHPs is that the number of parameters grows quadratically with the number of dimensions. Hence, with a dataset of moderate size, modeling each dimension of the process independently can lead to inefficiencies. Motivated by the observation that excitation functions often look similar across different dimensions, we propose a hierarchical modeling approach based on mixtures of nonparametric mixtures. More specifically, we adapt the approach introduced in \cite{Muller2004}, which models each of the excitation functions as a mixture of an idiosyncratic component and a common component shared by all excitation functions, and study some of the properties of such formulation. A second challenge in implementing MHP models is computational in nature. As is the case more generally, Markov chain Monte Carlo (MCMC) algorithms are the most common approach to computation for Bayesian models for Hawkes processes (e.g., see \citealp{Rasmussen2013}). 
However, MCMC algorithms for Hawkes process models are often too slow even for moderate sample sizes because, except for special cases such as the exponential excitation function, their complexity is quadratic in both the number of observations and the number of dimensions. This challenge is amplified in the case of  nonparametric models.  To address it, we expand on previous work on the use of stochastic gradient methods for MHPs (e.g., see \citealp{jiang2024improvements}), and develop a scalable stochastic variational inference (SVI) algorithm that can be used to fit our model.  An important part of this development involves a carefully comparison of the performance of SVI and MCMC methods, both in terms of accuracy and speed. {Furthermore, to illustrate the performance of the model, we applied our method to study the limit order book data for Amazon obtained from LOBSTER \citep{huang2011lobster}. Our analysis indicates that some of the excitation functions that arise from the analysis of this kind of data can have features such as non-monotonicity that are not captured by the kind of parametric forms that are commonly used in practice.}

%we consider an application to modeling high-frequency financial trading data inspired by the work of \cite{rambaldi2017role}.  

In summary, we make three key contributions in this paper:  (1) we propose a novel and flexible model for linear MHPs in which the various excitation functions are assigned a joint nonparametric prior that allows us to efficiently borrow information, (2) we develop MCMC and SVI algorithms for estimation and prediction in the context of this model, and thoroughly evaluate their relative performance, and (3) we illustrate the need for both nonparametric inference and fast computation in the context of a real-world application. 

The remainder of the paper is structured as follows. In Section 2 provides a brief review of multivariate Hawkes process. Section 3 outlines our model and discusses some of its properties. Section 4 describes the two computing algorithms we employ for our model, a Markov chain Monte Carlo algorithm and a stochastic gradient variational approximation.  Sections 5 and 6 discuss the results from simulations and real-world applications. Finally, Section 7 presents our conclusions and points out potential future directions for research.

\section{Multivariate Hawkes Processes}
\label{sec:model}

Let $N^{(1)}(t), \dots, N^{(K)}(t)$ be a collection of $K$ point processes defined on the positive real line $\mathbb{R}^{+}$, where  $N^{(k)}(t)$ represents the number of events on dimension $k$ that occur on the interval $[0,t]$. We denote a generic set of observations from this process by $\mathbf{X} = \{(t_i,d_i) : i = 1,\dots,n\}$, where $t_i \in \mathcal{R}^{+}$ represents the timestamp at which the $i$-th event occurs, and $d_i \in\{1, \ldots, K\}$ represents the dimension in which it occurs. Then, $\mathbf{X}$ follows a multivariate Hawkes process \citep{Hawkes1971, liniger2009multivariate} if the conditional intensity function on dimension $k$ has the following form:
$$
\lambda_{k}(t) \equiv \lim _{h \rightarrow 0} \frac{\mathbb{E}\left[N^{(k)}(t+h)-N^{(k)}(t) \mid \mathcal{M}_t\right]}{h}=\mu_{k}+\sum_{\ell=1}^K \sum_{\{ i: t_i<t, d_i=k \} } 
\alpha_{\ell,k} \tilde{\phi}_{ \ell,k}\left(t-t_i\right),
$$
where $\mathcal{M}_t$ denotes the subset of $\mathbf{X}$ for which $t_i < t$, $\mu_{k} > 0$ is the background intensity for dimension $k$, $\alpha_{\ell, k} > 0$ is the parameter that controls the strength by which past events from dimension $\ell$ influence the occurrence of new events on dimension $k$, and $\tilde{\phi}_{\ell,k}(\cdot): \mathbf{R}^+ \rightarrow \mathbf{R}^+$ is the (normalized) excitation function that controls how such influence decays over time. Note that we require that $\int_0^{\infty} \tilde{\phi}_{\ell,k,}(s)\text{d}s = 1$, which ensures that $\tilde{\phi}_{\ell,k}$ and $\alpha_{\ell,k}$ are identifiable.  The log-likelihood for the MHP can then be expressed as \citep{daley2008}:
\begin{equation}
\label{eq:obslik}
\begin{aligned}
\mathcal{L}\left(\mathbf{X} \mid \{ \mu_k \}, \{ \alpha_{k,\ell} \},  \{ \phi_{k,\ell} \}\right) &= \sum_{k=1}^K \sum_{d_i = k}  \log \lambda_{k}\left(t_{i}\right)- \sum_{k=1}^K\int_{0}^T \lambda_{k}(s) ds \\
&= \sum_{k=1}^K \sum_{d_i =k}  \log \left(\mu_k  +  \sum_{k=1}^K \sum_{\substack{\{ (i,j: j < i \\ d_{j} = k, d_{i} = \ell \} }}\alpha _{\ell,k}\tilde{\phi}_{\ell,k}(t_i - t_j)\right)- \sum_{k=1}^K \mu_k T  \\
& \;\;\;\;\;\;\;\;\;\;\;\;\;\;\;\;\;\;\;\;\;\;\;\;\;\;\;\;\;\;\;\;\;\;\;\;\;\;\;\;\;\;\; - \sum_{k=1}^K \sum_{\ell=1}^K \alpha _{\ell,k} \sum_{\{i:d_i =k \}} \tilde{\Phi}_{\ell,k}(T-t_i),
\end{aligned}    
\end{equation}
where %$\bfalpha$ and $\bfphi$ denote set of all $\alpha_{\ell,k}$s and all $\phi_{\ell,k}$s, respectively, and 
$\tilde{\Phi}_{\ell,k}(t) = \int_0^t \tilde{\phi}_{\ell,k}(s)ds$. 

An alternative construction of the MHP is as a multivariate branching process %as a recursive Poisson cluster 
in which the first generation of events in dimension $k$ (often called ``{immigrants}'' in the literature) arise from a homogeneous Poisson process with rate $\mu_{k}$, and the points in subsequent generations are generated from non-homogenous Poisson processes with rates given by the $\alpha_{\ell,k}$s and the interarrival times are controlled by the $\phi_{\ell,k}$s.  See \citealp{Hawkes1974} and \citealp{liniger2009multivariate} for details.  While the branching structure is typically latent (i.e., not observed), this construction both provides insight into the properties of the model and can be exploited to facilitate computation.  For example, this construction makes it clear that an MHP is stationary if and only if the spectral radius of the matrix $[\alpha_{\ell,k}]$ is less than 1.   In the sequel, we use the binary matrix $\bfB$ where 
\begin{align*}
B_{j,j} &=
\begin{cases}
1 & \quad \text{the $j$-th event is an immigrant} \\
0 & \quad \text{otherwise,} 
\end{cases}    \\
B_{i,j} &=
\begin{cases}
1 & \quad \text{the $j$-th event is an offspring of the $i$-th event} \\
0 & \quad \text{otherwise,}
\end{cases}   
\end{align*}
to encode the  latent branching structure associated with a realization of an MHP.  The augmented likelihood for the data is then given by
\begin{multline} \label{eq:complik}
\mathcal{L}\left(\mathbf{X}, \bfB \mid \{ \mu_k \}, \{ \alpha_{k,\ell} \},  \{ \phi_{k,\ell} \} \right) = \sum_{k=1}^K \sum_{\ell=1}^K\left[ \left|O _{k,\ell}\right|\left(\log \alpha _{k,\ell}\right)- \sum_{\substack{\{ (i,j) : i < j \\ d_{i} = k, d_{j} = \ell \}}} B_{i,j}\tilde{\phi}_{k,\ell}(t_j - t_i)\right] \\
+\sum_{\ell=1}^K\left|I_\ell\right| \log \mu_\ell- \sum_{\ell=1}^K \mu_\ell T - \sum_{k=1}^K \sum_{\ell=1}^K \alpha _{k,\ell} \sum_{\{i: d_i = k\}} \tilde{\Phi}_{k, \ell} (T - t_i),
\end{multline}
where $|I_{\ell}| = \sum_{d_i=\ell} I(B_{ii} = 1)$ 
and $|O_{k,\ell}| = \sum_{d_i = k, d_j = \ell, i<j} I(B_{ij} = 1)$ denote the number of immigrants for dimension $k$ and the number of offpring on dimension $k$ who arise from points on dimension $\ell$.

%note that given the branching structure $\boldsymbol{b}$, $|I_k|$ is the number of immigrants for dimension $k$ and $|O_{k,\ell}|$ is the number of descendants on dimension $k$ that are an offspring of an event on dimension $j$. Additionally, we let $\Tilde{\Phi}_{k,l}(\cdot)$ be the cumulative distribution function for $\tilde{\phi}_{k,l}(\cdot)$. 

\section{Nonparametric Bayesian modeling of excitation functions for multivariate Hawkes processes}
\label{sec:bnp}

As we discussed in the introduction, efficient and flexible estimation of the excitation functions is one of the key challenges when working with MHPs.  Because the normalized excitation functions are constrained to integrate out to one, this problem can be reduced to one of simultaneous estimation for a (finite) collection of densities.

The literature on Bayesian estimation for collections of densities is dominated by methods based on (dependent) nonparametric mixtures (e.g., see \cite{rodriguez2013nonparametric} and \citealp{quintana2022dependent}).  In the case of countable, exchangeable collections of  densities, examples of such dependent nonparametric mixture models include the Hierarchical Dirichlet process (HDP, \citealp{teh2004sharing}, the Nested Dirichlet process (NDP, \citealp{rodriguez2008nested}), mixture of mixture approaches such as those in \cite{Muller2004}, and the probit stick-breaking process \citep{rodriguez2011nonparametric}, among others. 

In the sequel, we adapt the methodology introduced in \cite{Muller2004} to the estimation of excitation functions in multivariate Hawkes processes.  More specifically, we consider mixtures of (scaled) Beta kernels of the form
\begin{align}
	\label{eq:dpm}
	\tilde{\phi}_{k,\ell}(t) = \int f_{\text{Beta}}(t \mid a,b, T_0) \text{d}G_{k,\ell}(a,b),
\end{align}
where $\{ G_{k,\ell} \}$ is a collection of almost-surely discrete mixing measures whose joint prior is described later in this section, and the kernel is given by
\begin{align*}
    f_{\text{Beta}}( \cdot \mid a, b, T_0) &:= \frac{\Gamma(a+b)}{\Gamma(a)\Gamma(b)} \frac{1}{T_0} \left(\frac{t}{T_0}\right)^{a-1}\left(1-\frac{t}{T_0}\right)^{b-1}, & a > 0, b > 0, 0 < t < T_0 .
\end{align*}
%a scaled Beta density function with support on $(0, T_0)$, indexed by shape parameters $a$ and $b$. 
The use of mixtures of Beta kernels allows the excitation function to have a flexible form, including non-monotonic decays. In fact, this type of mixtures have large support support of the space of absolutely continuous measures with bounded support on $[0,T_0]$ (see Section \ref{se:priorsupport}).  The fact that the kernel has compact support also helps computationally without severely impacting the ability of the model to capture key features of the excitation function.  Indeed, note that the compact support means that the number of observations pairs involved in the calculation of \eqref{eq:obslik} and \eqref{eq:complik} grows linearly rather than quadratically with $n$.  On the other hand, because the intensity function must eventually be strictly decreasing, in most cases a good approximation of intensity functions with infinite support can be obtained by choosing $T_0$ large enough.  A similar approach, sometimes called tapering, is often used to speed up computation for Gaussian process models (e.g., see \citealp{furrer2006covariance} and \citealp{kaufman2008covariance}). 

%We let $\text{G}_{k,\ell}(\cdot)$ be the mixing distribution for the two shape parameters that corresponds to $\tilde{\phi}_{k,\ell}$, which is a distribution on $\mathbb{R}^2$. Denote $\mathcal{M}(\mathbb{R}^2)$ as the space of all probability distributions on $\mathbb{R}^2$, we define a prior on $\text{G}_{k,\ell}$ over $\mathcal{M}(\mathbb{R}^2)$ using the following hierarchical model, following \cite{Muller2004}: 

Previous approaches to Bayesian nonparametric modeling of multivariate Hawkes process models (e.g., \citealp{donnet2020nonparametric} and \citealp{Markwick2020}) have relied on independent priors on each of the mixing distributions $G_{k,\ell}$.  Instead, we propose to propose to model them jointly using further mixtures of the form
\begin{align*}
	G_{k,\ell}(\cdot)&=\varepsilon H_0(\cdot)+(1-\varepsilon) H_{k,\ell}(\cdot), & k,\ell &= 1,\dots,K,
 %\text{H}_0, \text{H}_{k,\ell} &\sim \text{DP}(\alpha_{\text{DP}}, \operatorname{Gamma}\left(c_a, d_a\right) \times \operatorname{Gamma}\left(c_b, d_b\right)) ~~ k,\ell = 1,\dots,K,  \\
 %\varepsilon &\sim \text{Beta}(1,1).
\end{align*}
where $0 \le \varepsilon \le 1$, and $H_0$ and the various $H_{k,\ell}$s are in turn given independent priors as described below.  We can think of $H_0$ as a ``common component'' that captures features shared across all dimension pairs, and of $H_{k, \ell}$ as an ``idiosyncratic'' component that captures features that are specific to each.  The parameter $\varepsilon$ controls the relative weight associated with each of these two components and, therefore, the level of dependence across the $G_{k,\ell}$s.  Indeed, note that setting $\varepsilon = 0$ corresponds to assigning mutually independent priors to the $G_{k,\ell}$s, whereas $\varepsilon = 1$ implies that the intensity functions are exactly identical (perfect dependence) across all pairs. In practice, we expect  $\varepsilon$ to lie somewhere in between these two extremes, so we treat it as an unknown parameter that needs to be estimated from the data.

We model the common and the idiosyncratic components as independent realization from  Dirichlet process priors, so that
\begin{align*}
  H_0(\cdot) &= \sum_{h=1}^{\infty} w^{0}_{h} \delta_{\left( a^{0}_{h}, b^{0}_{h} \right)}  & 
  H_{k,\ell}(\cdot) &= \sum_{h=1}^{\infty} w^{k,\ell}_{h} \delta_{\left( a^{k,\ell}_{h}, b^{k,\ell}_{h} \right)}  
\end{align*}
where $\delta_{\boldsymbol{\theta}}(\cdot)$ denotes a point mass at $\boldsymbol{\theta}$,
$a_h^{0} \sim \text{Gamma}\left(c_{a}^{C}, d_{a}^{C}\right)$, $b_h^{0} \sim \text{Gamma}\left(c_{b}^{C}, d_{b}^{C}\right)$, $a_h^{k,\ell} \sim \text{Gamma}\left(c_{a}^{I}, d_{a}^{I}\right)$ and $b_h^{k,\ell} \sim \text{Gamma}\left(c_{b}^{I}, d_{b}^{I}\right)$ independently for all $h=1,2, \ldots$ and $k,\ell = 1,\ldots, K$, and $w^{0}_{h} = v^{0}_{h} \prod_{r < h} (1 - v^{0}_{h})$ and $w^{k,\ell}_{h} = v^{k,\ell}_{h} \prod_{r < h} (1 - v^{k,\ell}_{h})$, with $v^{0}_{h} \sim \text{Beta}(\gamma,1)$ and $v^{k,\ell}_{h} \sim \text{Beta}(\gamma,1)$ also independent for  all $h=1,2, \ldots$ and $k,\ell = 1,\ldots, K$.  Allowing different parameters for the centering measure of the common and idiosyncratic components allows us to potentially encode relevant prior information.  For example, it is common to assume that excitation functions are roughly decreasing, so we might want to reflect that in the shape of the common component by favoring values of $a_h^{0} \ll  1$ and $b_h^{0} \gg 1$.  On the other hand, we might expect idiosyncratic components to reflect deviations from monotonicity, so we might want their centering measure to favor less extreme values for $a_{h}^{{k,\ell}}$ and $b_{h}^{k,\ell}$.  Finally, and in a similar spirit, we model the mixture weight $\epsilon$ using a uniform distribution on $[0,1]$.

It is worthwhile noting that the model in \cite{Muller2004} has sometimes been criticized for being overparameterized.  However, 
unlike other applications (e.g., see \cite{Rodriguez2017}),
we are interested in estimating the aggregate excitation functions $\{ \tilde{\phi}_{k,l}(t) \}$ defined in \eqref{eq:dpm}, which are themselves identifiable.  This alleviates any concerns about overparameterization in this particular context.

The model is completed by eliciting priors for the parameters $\mu_1, \ldots, \mu_K$ and the coefficients $\{ \alpha_{k, \ell} \}$ and the concentration $\gamma$.  In particular, we set 
$\mu_{\ell}   \sim %\stackrel{i . i . d}{\sim} 
\operatorname{Gamma}\left(e_{\ell}, f_{\ell}\right)$ and $\alpha_{k,\ell} \sim %\stackrel{i . i . d}{\sim} 
\operatorname{Gamma}\left(g_{k,\ell}, h_{k,\ell}\right)$ independently for all $k,\ell = 1, \dots, K$.  %A more detailed discussion of the choice of hyperparameters is provided in Section \ref{se:illustration} in the context of our illustration.

\subsection{Prior support}\label{se:priorsupport}

An important consideration when designing nonparametric priors is their associated support (e.g., see \citealp{walker2004priors}).  This question has been well studied for nonparametric mixture priors placed on a single unknown distribution, and they naturally extend to situations in which independent priors are assigned to individual members of a countable collection of distribution.  However, results for dependent priors on collections of distribution have been less studied.

In order to show that our nonparametric prior has large  Kullback–Leibler support, we discuss first a simpler result that applies to the individual priors we place on the common and idiosyncratic components:

\begin{theorem}\label{th:support}
Let $\mathcal{H}$ be the set of all absolutely continuous densities on $[0,T_0]$, and $\Pi$ be the prior on $\mathcal{H}_0$ induced by a Dirichlet process mixture of the form
%\begin{align*}
$	 \int f_{\text{Beta}}(t \mid a,b, T_0) \text{d}H(a,b)$, 
%\end{align*}
where $H \sim \text{DP}(\gamma, \operatorname{Gamma}\left(c_a, d_a\right) \times \operatorname{Gamma}\left(c_b, d_b\right))$.  Then, $\Pi$ has full Kullback-Leibler support on $\mathcal{H}$, i.e., for any $\phi^{0} \in \mathcal{H}$ and $\epsilon > 0$ we have
$$
\Pi\left( \left\{ \phi : D_{KL}(\phi^{0} \, || \, \phi) < \epsilon \right\}   \right) > 0,
$$
where
$D_{KL}(\phi^{0} \, || \, \phi) = \int \phi^{0}(s) \log (\phi^{0}(s)/\phi(s)) \dd s$ is  the Kullback-Leibler divergence between $\phi^{0}$ and $\phi$.

%where $T > 0$,
% and let $f_0(x) \in \mathcal{H}_0$.  

%If $f_0(x) \in \mathcal{H}_0$, $\Pi$ is the prior induced by the scaled Beta likelihood mixture kernel on [0,T] and $\text{DP}(\gamma, \operatorname{Gamma}\left(c_a, d_a\right) \times \operatorname{Gamma}\left(c_b, d_b\right))$ on $\mathcal{H}_0$. Then $f_0$ is in the KL support of $\Pi$, i.e. there exists $\epsilon > 0$ such that $\Pi(\{g: D(f_0, g)\}) < \epsilon$, where $D(f_0 || g) = \int f_0 \log (f_0/g)$ is defined as the Kullback-Leibler distance between $f_0$ and $g$.
\end{theorem}
The proof of the this theorem is a direct extension from the results in \citep{wu2008kullback}, and is available in supplemental material Section A. This result is the basis for the following corollary, which extends the result to the setting of the our dependent model 

\begin{corollary}\label{cor:support}Consider the joint prior $\Pi^{*}$ on $\mathcal{H}^{*} = \oplus_{k=1}^{K\times K} \mathcal{H}$ induced by the nonparametric mixture of mixtures introduced in Section \ref{sec:bnp}.
%
%
%$K \times K$ absolutely continuous densities on $[0,T_0]$ $f^{1,1}_0, \dots, f^{K,K}_0 \in \mathcal{H}_0$. Consider the following joint prior $\Pi^* : \oplus_{k=1}^{K\times K} \mathcal{H}_0 \rightarrow \mathbb{R}$ for $(f^{1,1}, \dots, f^{K,K})$ such that 
%\begin{equation*}
%\begin{aligned}
%    f^{k,\ell} &= \varepsilon g^0 + (1-\varepsilon) g^{k,\ell},~~ k, \ell = 1, \dots, K, \\
%    g^{0}(x) &= \int_{\mathbb{R}^2} f_{\operatorname{Beta}}(x \mid a, b) dP_O(a,b), \\    
%    g^{k,\ell}(x) &= \int_{\mathbb{R}^2} f_{\operatorname{Beta}}(x \mid a, b) dP_I^{k,\ell}(a,b),~~ k = 1, \dots, K, \ell = 1, \dots, K \\
%    P_O &\sim \operatorname{DP}(\gamma, \operatorname{Gamma}(c_a^O, d_a^O)\times \operatorname{Gamma}(c_b^O, d_b^O)), \\
%   P_I^{k,\ell} &\stackrel{i.i.d.}{\sim} \operatorname{DP}(\gamma, \operatorname{Gamma}(c_a^I, d_a^I)\times \operatorname{Gamma}(c_b^I, d_b^I)), \\
%    \varepsilon &\sim \operatorname{Dirichlet}(1,1),
%\end{aligned}
%\end{equation*}
Then, for any collection $\{ \phi_{1,1}^{0}, \ldots, \phi_{K,K}^{0} \} \in \mathcal{H}^{*}$ and any $\epsilon > 0$ we have 
\begin{align*}
    \Pi^*\left( \left\{\phi_{1,1}, \dots, \phi_{K,K}: D_{KL}(\phi^0_{k,\ell}  \, || \,  \phi_{k,\ell}) < \epsilon \mbox{ for all } k, \ell = 1, \dots, K \right\}\right) > 0.
\end{align*}
\end{corollary}
The proof of the this theorem is also available in supplemental material Section A.

\section{Computation}\label{se:computation}

This Section describes two alternative computational strategies for inference on the posterior distribution associated with the model described in Section \ref{sec:model}.  The first strategy is based on a Markov chain Monte Carlo (MCMC) algorithm that heavily relies on Gibbs sampling steps.  This algorithm is relatively simple to implement and can provide accurate estimates of most posterior quantities of interest, but it is comparatively slow and impractical for large datasets.  The second strategy we discuss is based on a stochastic gradient variational approximation to the posterior distribution.  The resulting algorithm is orders of magnitudes faster than MCMC and tends to yield reasonably accurate point estimates, but it also tends to underestimates the uncertainty associated with the posterior distribution.

Both of the computational strategies we discuss rely on a common set of approximations and latent variable augmentations, which we describe before getting into the details of each of them.  Firstly, both algorithms rely on (an approximation to) the full data likelihood in \eqref{eq:complik}.  Augmenting the model with the (latent) matrix $\mathbf{B}$ is what enables us to treat the problem of estimating the excitation function as one of density estimation, which in turn serves as an additional motivation for the use of mixture models in Section \ref{sec:model}.  The main challenge with this approach is that \eqref{eq:complik} is not fully tractable.  In particular, note that \eqref{eq:complik} includes a term of the form
\begin{align}
\Upsilon =	\sum_{\ell=1}^{K} \int_{0}^T \lambda_{\ell}(s) d s =\sum_{\ell=1}^K\mu_\ell T +  \sum_{k=1}^K\sum_{\ell=1}^K  \alpha_{k,\ell}\sum_{d_i = k} \tilde{\Phi}_{k,\ell}(T - t_i) .
\end{align}
The term $\int_0^T \lambda_{\ell}(s) d s$ is often called the \textit{compensator} for the conditional density function $\lambda_{\ell}(t)$, and it captures the likelihood of there being infinitely many none-events between observations on the temporal domain $[0,T_0]$. The complex structure associated with the compensator can a major hurdle in the development of computational algorithms for Hawkes process models (e.g., see  \citealp{jiang2024improvements}).  A common solution, which we adopt in this paper, is the approximation suggested in \cite{lewis2011nonparametric}:
\begin{align}\label{eq:approxcomp}
\alpha_{k,\ell}\tilde{\Phi}_{k,\ell}(T - t_i) \approx \alpha_{k,\ell}.    
\end{align}
This approximation is particularly appealing in our context because we assume that $\phi_{k,\ell}(\cdot)$ has bounded support over $[0,T_0]$.  Therefore, the approximation will only affect the likelihood evaluation for a relative small portion of the whole observations. 

Secondly, for the purposes of computational implementation, we consider a finite truncation approximation on the number of component in our model that relies on the ideas of \citep{Ishwaran2002}.  More specifically, we consider the finite mixture model of the form
\begin{align*}
	\tilde{\phi}_{k,\ell}(t) &=  \varepsilon \sum_{h=1}^H p^{0}_h f_{\text{Beta}}(t \mid a^{0}_h,b^{0}_h) + (1 - \varepsilon)\sum_{h=1}^H p^{k,\ell}_h f_{\text{Beta}}(t \mid a^{k,\ell}_h,b^{k,\ell}_h), %& k,\ell = 1,\ldots,K,
\end{align*}
where
\begin{align*}
    \bfp^{0} = (p^{0}_1, \dots, p^{0}_H) &\sim \text{Dirichlet} \left(\frac{\gamma}{H}, \dots, \frac{\gamma}{H} \right) , &
    \bfp^{k,l} = (p^{k,\ell}_1, \dots, p^{k,\ell}_H) &\sim \text{Dirichlet} \left(\frac{\gamma}{H}, \dots, \frac{\gamma}{H} \right) . 
\end{align*}
for all $k, \ell = 1, \dots, K$.  The truncation level $H$ is an upper bound on the number of mixture components.  As long as $H$ is chosen to be large enough, this kind of (overfitted) finite-dimensional mixture provides an excellent approximation to the its nonparametric counterpart, one that has been repeatedly exploited to design computational algorithms for non-parametric mixture models.

Finally, and as is common in the context of mixture modeling, we introduce various sets of latent allocation variables, indicating the mixture component each observation comes from. Since \eqref{eq:complik} is evaluated on inter-arrival times across dimensions, we define $$\mathcal{T}_{k,\ell} = \{(i,j) :d_i = k, d_j = \ell,  i < j\}$$ as the set of all tuples of observation indices corresponding to transitions between a potential parent in dimension $i$ and a potential offspring in dimension $j$. Note that we only evaluate the transitions indicated by the branching structure, i.e. $B_{ij} = 1$. For each tuple $(i,j) \in \mathcal{T}_{k, \ell}$, we denote $(W_{i,j},Z_{i,j})$ as a set of latent variables that jointly determine the mixture component that $t_j - t_i$ belongs to. More specifically,
$$
t_j - t_i \mid Z_{i,j}, W_{i,j} \sim \left\{\begin{array}{lc}
\text{Beta}\left(a_{Z_{i,j}}^0, b_{Z_{i,j}}^0\right) & \text{if }W_{i,j}=1 \\
\text{Beta}\left(a^{k,\ell}_{Z_{i,j}}, b^{k,\ell}_{Z_{i,j}}\right) &  \text{if }W_{i,j}=0
\end{array}.\right.
$$
Given $\boldsymbol{p}^{0}$, $\{ \boldsymbol{p}^{k,\ell} \}$ and $\epsilon$, $\{ W_{i,j}^{k, \ell} \}$ and $\{ Z_{i,j}^{k, \ell} \}$ haver a joint distribution given by:
\begin{align*}
P(Z_{i,j} = h , W_{i,j} = 1 \mid \boldsymbol{p}^{0}, \varepsilon) &= \varepsilon p_h^0,\\  
P(Z_{i,j} = h , W_{i,j} = 0 \mid \{ \boldsymbol{p}^{k,\ell} \}, \varepsilon) &= (1-\varepsilon)p_h^{k,\ell}.
\end{align*}
%Note that $P(W_{i,j}, Z_{i,j} \mid \boldsymbol{p}, \varepsilon)$ corresponds to the probabilities that $t_j-t_i$ belongs to the common component or the idiosyncratic component. 

\subsection{Markov chain Monte Carlo algorithm}
\label{sec:poscomp}

Once the approximations and data augmentation discussed above are introduced, the posterior distribution takes the form 
\begin{multline}\label{eq:approx_posterior}
f\left( \{ \mu_k \}, \{ \alpha_{k,\ell} \},  \{ W_{i,j} \},  \{ Z_{i,j} \}, \{ p_h^{0} \}, \{ a_{h}^{0} \}, \{ b_{h}^{0} \}, \{ p_h^{k,\ell} \}, \{ a_{h}^{k,\ell} \}, \{ b_{h}^{k,\ell}\}, \varepsilon, \bfB \mid \bfX \right) \propto \\
f^{*} \left( \bfX, \bfB \mid \{ \mu_k \}, \{ \alpha_{k,\ell} \}, \{ W_{i,j} \},  \{ Z_{i,j} \}, \{ a_{h}^{0} \}, \{ b_{h}^{0} \}, \{ a_{h}^{k,\ell} \}, \{ b_{h}^{k,\ell}\}\right)  \\
f\left(\{ \mu_k \}, \{ \alpha_{k,\ell} \}, \{ W_{i,j} \},  \{ Z_{i,j} \}, \{ a_{h}^{0} \}, \{ b_{h}^{0} \}, \{ a_{h}^{k,\ell} \}, \{ b_{h}^{k,\ell}\}\right)
%
%
%
%
%
%f\left( \{ W_{i,j} \} \mid \varepsilon \right) \\
%f\left( \{ Z_{i,j} \} \mid \{ W_{i,j} \} , \{ p_h^{0} \}, \{ p_h^{k,\ell}, \varepsilon \}\right) ) f\left(\{ \mu_{k} \}\right) f\left(\{ p_h^{0} \} \right) f\left(\{ a_{h}^{0} \}\right) f\left(\{ b_{h}^{0} \}\right) f\left( \{ p_h^{k,\ell} \} \right) \\
%
%f\left(\{ a_{h}^{k,\ell} \}\right) f\left( \{ b_{h}^{k,\ell}\}\right) f\left(\varepsilon\right)
\end{multline}
where 
\begin{multline*}%\label{se:postcomp}
\log f^{*} \left( \bfX, \bfB \mid  \{ \mu_k \}, \{ \alpha_{k,\ell} \}, \{ W_{i,j} \},  \{ Z_{i,j} \}, \{ a_{h}^{0} \}, \{ b_{h}^{0} \}, \{ a_{h}^{k,\ell} \}, \{ b_{h}^{k,\ell}\}\right)  = \sum_{k=1}^K \sum_{\ell=1}^K \left|O _{k,\ell}\right|\left(\log \alpha _{k,\ell}\right) \\ 
% +\sum_{\ell=1}^K\left|I_\ell\right| \log \mu_\ell- \sum_{\ell=1}^K \mu_\ell T -\sum_{k=1}^K \sum_{\ell=1}^K \Bigg[ 
\sum_{\substack{\{ (i,j) : i < j \\ d_{i} = k, d_{j} = \ell \}}} B_{i,j}  I(Z_{i,j} = h) 
\bigg[
W_{i,j}\log f_{\text{Beta}} \left(t_j - t_i \mid a_h^{k,\ell}, b_h^{k,\ell} \right) \\
+ (1-W_{i,j})\log f_{\text{Beta}} \left(t_j - t_i \mid a_h^{0}, b_h^{0} \right) \bigg] 
-\alpha _{k,\ell} n_k 
\Bigg] ,
% \sum_{\{i: d_i = k\}} \tilde{\Phi}_{k, \ell} (T - t_i) 
\end{multline*}
with $n_k = \sum_{i = 1}^N I(d_i = k)$, is the approximate log-likelihood based on \eqref{eq:approxcomp}, and the joint prior introduced above defines $f\left(\{ \mu_k \}, \{ \alpha_{k,\ell} \}, \{ W_{i,j} \},  \{ Z_{i,j} \}, \{ a_{h}^{0} \}, \{ b_{h}^{0} \}, \{ a_{h}^{k,\ell} \}, \{ b_{h}^{k,\ell}\}\right)$ .

Many of the full conditional distributions associated with \eqref{eq:approx_posterior} belong to known families of distributions and are therefore easy to sample from.  For example, the full conditional posteriors for each of the $\mu_k$s and $\alpha_{k,\ell}$s correspond to independent Gamma distributions.  Similarly the branching structure $\bfB$ and the indicators $\{ Z_{i,j} \}$ and $\{ W_{i,j} \}$ are all categorical variables and sampling from their full conditional distributions is straightforward.  Finally, the full conditional posterior distribution for $\bfp^{0}$ and each of the $\bfp^{k,l}$s correspond to Dirichlet distributions on the $H$ dimensional simplex, and the full conditional for $\varepsilon$ follows an updated Beta distribution.  The key exception are the parameters $\{ a_{h}^{0} \}$, $\{ b_{h}^{0} \}$, $\{ a_{h}^{k,\ell} \}$ and $\{ b_{h}^{k,\ell}\}$.   We propose to update these though a random walk Metropolis-within-Gibbs steps on their log scales.  Details of the algorithm can be seen in supplemental material Section B.%, and a full implementation is available at \url{http:....}

\subsection{Stochastic variational inference}
\label{sec:svi}

In variational inference, the intractable posterior distribution $f(\bftheta \mid \bfX)$ is replaced with a tractable approximation $q_{\hat{\bfeta}}(\bftheta)$ obtained by minimizing the Kullback–Leibler divergence between it and  a member of a carefully chosen parametric family $\{ q_{\bfeta}(\bftheta) : \bfeta \in H\}$, i.e., 
$$
\tilde{\boldsymbol{\eta}}=\arg \max _{\boldsymbol{\eta} \in H} \mathrm{E}_{q_{\boldsymbol{\eta}}} \log \left\{\frac{p(\boldsymbol{\theta}, \mathbf{X})}{q_{\boldsymbol{\eta}}(\boldsymbol{\theta})}\right\} := \arg \max _{\boldsymbol{\eta} \in H} \operatorname{ELBO}_{\boldsymbol{\eta}}.
$$

In the sequel, we work with a mean-field variational approximation that assumes independence across all parameters, except for the $(W_{i,j}, Z_{i,j})$ pairs,
\begin{multline}\label{eq:mf}
   q_{\boldsymbol{\eta}}\left( \{ \mu_k \}, \{ \alpha_{k,\ell} \},  \{ W_{i,j} \},  \{ Z_{i,j} \}, \{ p_h^{0} \}, \{ a_{h}^{0} \}, \{ b_{h}^{0} \}, \{ p_h^{k,\ell} \}, \{ a_{h}^{k,\ell} \}, \{ b_{h}^{k,\ell}\}, \varepsilon, \bfB \mid \bfX \right) = \\
   \prod_{\ell=1}^K q_{\eta_{\mu_\ell}}(\mu_\ell)\prod_{k=1}^K\prod_{\ell=1}^K q_{\eta_{\alpha_{k,\ell}}}(\alpha_{k,\ell}) \prod_{h=1}^{L_0} q_{\eta_{a_h^0}}(a_h^0) q_{\eta_{b_h^0}} (b_h^0) q_{\eta_{p_h^0}}(p_h^0) \prod_{k=1}^K\prod_{\ell=1}^K \prod_{h=1}^{L} q_{\eta_{a_h^{k,\ell}}}(a_h^{k,\ell}) q_{\eta_{b_h^{k,\ell}}} (b_h^{k,\ell}) q_{\eta_{p_h^{k,\ell}}}(p_h^{k,\ell}) \\ 
   \prod_{j=1}^N \prod_{i=1}^{j} q_{\eta_{B_{ij}}} \left(B_{ij}\right)
   %s
   \prod_{k=1}^{K} \prod_{\ell=1}^{K} \prod_{(i,j) \in \mathcal{T}_{k,l}}
   q_{\eta_{W_{i,j}}, \eta_{Z_{i,j}}}\left(W_{i,j}, Z_{i,j} \right)
   q_{\eta_{\varepsilon}}(\varepsilon) .
\end{multline}
The families to which each of the individual terms in the variational approximation belong match those of the corresponding priors, facilitating computation (e.g., see \citealp{Blei06}).  Indeed, similarly to the MCMC algorithm, the conditional conjugacy of many of the priors and the choice of the approximation family means that most of the parameters can be easily optimized using an iterative coordinate descent algorithm that alternates updates of each of the parameters.  Again, the one exception are the parameters $\{ a_{h}^{0} \}$, $\{ b_{h}^{0} \}$, $\{ a_{h}^{k,\ell} \}$ and $\{ b_{h}^{k,\ell}\}$.  To update them, we rely on the approximate lower bound for the ELBO %of Bayesian finite mixture models with Beta kernels 
that was introduced in \cite{ma2011bayesian}. More specifically, we approximate $\mathrm{E}_{a, b}\left[\log \frac{\Gamma(a+b)}{\Gamma(a) \Gamma(b)}\right]$ using a  Taylor's expansion:  
\begin{equation}
\begin{aligned}
\label{eq:elbolb}
\operatorname{E}_{a,b} \left[\log \frac{\Gamma\left(a+b\right)}{\Gamma\left(a\right) \Gamma\left(b\right)} \right] &\geq \log \frac{\Gamma\left(\bar{a}+\bar{b}\right)}{\Gamma\left(\bar{a}\right) \Gamma\left(\bar{b}\right)} + \bar{a}[\psi(\bar{a}+\bar{b}) - \psi(\bar{a})](\operatorname{E}[\log a] - \log \bar{a}) \\&+\bar{b}[\psi(\bar{a}+\bar{b})-\psi(\bar{b})](\operatorname{E}[\log b]-\log \bar{b})\\
&+ \frac{1}{2}\bar{a}^2\left[\psi^{\prime}(\bar{a}+\bar{b})-\psi^{\prime}(\bar{a})\right] \operatorname{E}\left[(\log a-\log \bar{a})^2\right] \\&+ \frac{1}{2}\bar{b}^2\left[\psi^{\prime}(\bar{a}+\bar{b})-\psi^{\prime}(\bar{b})\right] \operatorname{E}\left[(\log b-\log \bar{b})^2\right] \\
&+\bar{a} \cdot \bar{b} \cdot \psi^{\prime}(\bar{a}+\bar{b})(\operatorname{E}[\log a]-\log \bar{b})(\mathbf{E}[\log b]-\log \bar{b}) .
\end{aligned}
\end{equation}
where $\psi(z)=\frac{\mathrm{d}}{\mathrm{~d} z} \ln \Gamma(z)=\frac{\Gamma^{\prime}(z)}{\Gamma(z)}$ denotes the Digamma function \citep{abramowitz1964handbook}, and $\bar{a}, \bar{b}$ are the expectations of $a$ and $b$ under the variational Gamma distribution, e.g. $\bar{a} = \frac{\eta_{a,1}}{\eta_{a,2}}$ for $a \sim \text{Gamma}(\eta_{a,1}, \eta_{a,2})$.  The use of this lower bound leads to simple, closed form updates for $\bfeta_{a^{0}}$s, $\bfeta_{b^{0}}$s, $\bfeta_{a^{k,\ell}}$s and $\bfeta_{b^{k,\ell}}$s.

Finally, to further speed up computation, we implement a stochastic gradient version of the conjugate descent algorithm that closely follows the ideas of \cite{Hoffman2013} (see \citealp{jiang2024improvements} as well).  At each iteration, the algorithm updates the allocation and branching variables $W_{i,j}$, $Z_{i,j}$ and $B_{i,j}$ for the observations contained in a randomly chosen segment of length $\kappa$ within the interval $[0,T]$ before updating the remainder of the model parameters.  The learning rate we use for this update is given by $\rho_s = \rho_0 (s + \tau_1)^{\tau_2}$, where $\rho_0$ is a common scaling factor, $\tau_1 \geq 0$ is the delay parameter that `slows down' early iterations and $\tau_2 \in (0.5,1]$ is the forgetting rate that controls the rate of the exponential decay. Further details of this algorithm can be seen in supplemental material Section C.%, and a full implementation is also available at \url{https://github.com/AlexJiang1125/MHPDDP}.

\section{Simulation Studies}
\label{se:sim}

\paragraph{Data generating mechanisms.} 

The data is generated from a multivariate Hawkes process with $K = 2$ dimensions. All scenarios we consider share the following parameter values: $\boldsymbol{\alpha} = \begin{bmatrix}
        0.6 & 0.15 \\ 0.3 & 0.6
    \end{bmatrix}$, 
    $\boldsymbol{\mu} = \begin{bmatrix}
        0.05 \\ 0.1
    \end{bmatrix}$ and $T = 15000$.  For the triggering kernels, we consider two scenarios based on whether the Beta mixture component is correctly specified. First, we consider scenarios where the excitation functions kernels are correctly-specified as mixtures of Beta distributions: 
\begin{align*}
    \tilde{\phi}^{\text{true}}_{k,\ell}(t) = \varepsilon_{\text{true}} \text{Beta}(t \mid a_{\text{true}}^0, b_{\text{true}}^0, T_0) + (1-\varepsilon_{\text{true}}) \text{Beta}(t \mid a_{\text{true}}^{k,\ell}, b_{\text{true}}^{k,\ell}, T_0) 
\end{align*}
where 
\begin{align*}
    a_{\text{true}}^0 &= 1, & b_{\text{true}}^0 &= 4, & \begin{bmatrix}
        a_{\text{true}}^{11} & a_{\text{true}}^{12} \\ a_{\text{true}}^{21} & a_{\text{true}}^{22}
    \end{bmatrix} &= \begin{bmatrix}
        2 & 4 \\ 1.5 & 1
    \end{bmatrix}, & \begin{bmatrix}
        b_{\text{true}}^{11} & b_{\text{true}}^{12} \\ b_{\text{true}}^{21} & b_{\text{true}}^{22}
    \end{bmatrix} &= \begin{bmatrix}
        6 & 1 \\ 5 & 1
    \end{bmatrix}, & T_0 = 1.
\end{align*}
We also consider a misspecified model with exponential excitation functions : 
\begin{align*}
    \tilde{\phi}^{\text{true}}_{k,\ell}(t) = \varepsilon_{\text{true}} \exp (-t)+ (1-\varepsilon_{\text{true}})  \exp (-\lambda_{\text{true}}^{j,k} t), 
\end{align*}
where $\begin{bmatrix}
    \lambda_{\text{true}}^{11} & \lambda_{\text{true}}^{12} \\ \lambda_{\text{true}}^{21} & \lambda_{\text{true}}^{22} 
\end{bmatrix} = \begin{bmatrix}
    2 & 0.8 \\ 0.8 & 2
\end{bmatrix}$.  Additionally, for each of these two scenarios, we consider five values of $\varepsilon_{\text{true}}$, $\{ 0, 0.2, 0.5, 0.8, 1\}$.

\paragraph{Benchmark methods.} We fit our model using the two computational approaches discussed in Section \ref{se:computation}.  We also compare our method where $\varepsilon$ is treated as an unknown parameter (`RANDOM'), to two benchmark versions, one where there is no information borrowing (`IDIO', which corresponds to fixed $\varepsilon = 0$) or the triggering kernels are identical (`COMMON', which corresponds to fixed $\varepsilon = 1$). Each of IDIO and COMMON are in turn fitted using versions of both the MCMC and the SGVI algorithms discussed in Section \ref{se:computation}. 
Finally, we consider a frequentist benchmark method based on piecewise basis kernels using the EM algorithm (EM-BK, see \citealp{zhou2013learning}).

\paragraph{Performance metrics.}  %We evaluate both the point and uncertainty estimation accuracy for the triggering kernels, using a set of performance metrics.
We use the root mean integrated squared error (RMISE) as a metric for point estimation accuracy:
\begin{align*}
    \text{RMISE}(\boldsymbol{\phi}) &= \frac{1}{K^2} \sum_{k=1}^K \sum_{\ell=1}^K \sqrt{\int_0^{+\infty}\left(\phi_{k, \ell}^{\text {true }}(x)-\hat{\phi}_{k, \ell}(x)\right)^2 \mathrm{~d} x} .
\end{align*}
%where $\text{RMISE}(\boldsymbol{\phi})$ can be interpreted as 
This is just the $L_2$ distance between the excitation functions and its estimate (which, for Bayesian procedures, corresponds to the posterior mean), averaged over all dimension pairs.  In practice, we approximate the integral involved in the definition of $\text{RMISE}(\boldsymbol{\phi})$ by averaging the value of the function over a fine grid. %posterior samples on grid values:
%\begin{align*}
%    \hat{\text{RMISE}}(\boldsymbol{\phi}) &= \frac{1}{B K^2} \sum_{b=1}^B \sum_{k=1}^K \sum_{\ell=1}^K \sqrt{ \frac{h}{|G|} \sum_{x_g \in G} \left(\phi_{k, \ell}^{\text {true }}(x_g)-\hat{\phi}_{k, \ell}^b(x_g)\right)^2 \mathrm{~d} x}
%\end{align*}
%where $G = \{h, 2h, 3h, \dots\}$ represents the set of grid points, and $\hat{\phi}_{k, \ell}^b(x_g)$ is the estimated triggering kernel evaluated at $x_g$ where the parameters $\boldsymbol{a}^0, \boldsymbol{b}^0, \boldsymbol{w}^0, \boldsymbol{a}^{k,\ell}, \boldsymbol{b}^{k,\ell}, \boldsymbol{w}^{k,\ell}$ are replaced with the posterior samples for the MCMC method and the samples generated from the variational distribution for the SVI method.

On the other hand, we use the average coverage rate (ACR) and the interval score (IS) to evaluate uncertainty estimation. ACR is defined as the proportion of correct coverages of the triggering kernels evaluated on the grid points, averaged over all the dimension pairs. IS \citep{gneiting2007strictly} is a generalization of ACR that penalizes wider interval lengths and low coverage rates.

\paragraph{Results.} Tables \ref{tab:pe-rmise} through \ref{tab:ue-is} show the RMISE, coverage rates and interval scores across five different values of $\varepsilon_{\text{true}}$ under the correctly-specified data generating mechanism for the various approaches under consideration. Similarly, Tables \ref{tab:pe-rmise2} through \ref{tab:ue-is2} (shown in supplemental materials Section D) show the results under the mis-specified data generating mechanism. For the MCMC algorithm under the correctly-specified scenario, `IDIO' dominates the other two methods when $(\varepsilon_{\text{true}} = 0)$, as it has the lowest RMISE. Likewise, `COMMON' dominates under the scenario where $(\varepsilon_{\text{true}} = 1)$.  This is not surprising, as in both cases the best performing model matches the true model.  However, for scenarios where the true data is a mixture of a common and an idiosyncratic components, our method (`RANDOM') shows the lowest RMISE. The ranking of the models under stochastic variational inference is similar, but we see slightly higher RIMSE values across the board. Finally, RMISE values for EM-BK are much higher than those for the MCMC algorithm in all scenarios. The `IDIO' method for MCMC is closer to the nominal coverages in most scenarios, but have slightly lower interval scores due to having wider intervals. The SVI algorithm tends to underestimate parameter uncertainty, causing much lower coverages and higher interval scores. The results for the mis-specified case are similar.

\begin{table}[]
\centering
\begin{tabular}{cccccccc}
\hline
 &
  \multicolumn{3}{c}{MCMC} &
  \multicolumn{3}{c}{SVI} &
  EM-BK \\ 
 $\varepsilon_{\text{true}}$ &
  RANDOM &
  IDIO &
  COMMON &
  RANDOM &
  IDIO &
  COMMON &
  - \\ \hline
$0$ &
  \begin{tabular}[c]{@{}c@{}}0.097\\ (0.016)\end{tabular} &
  \begin{tabular}[c]{@{}c@{}}0.096\\ (0.021)\end{tabular} &
  \begin{tabular}[c]{@{}c@{}}0.775\\ (0.001)\end{tabular} &
  \begin{tabular}[c]{@{}c@{}}0.279\\ (0.032)\end{tabular} &
  \begin{tabular}[c]{@{}c@{}}0.182\\ (0.129)\end{tabular} &
  \begin{tabular}[c]{@{}c@{}}0.859\\ (0.011)\end{tabular} &
  \begin{tabular}[c]{@{}c@{}}0.828\\ (0.297)\end{tabular} \\
$0.2$ &
  \begin{tabular}[c]{@{}c@{}}0.102\\ (0.014)\end{tabular} &
  \begin{tabular}[c]{@{}c@{}}0.102\\ (0.012)\end{tabular} &
  \begin{tabular}[c]{@{}c@{}}0.621\\ (0.002)\end{tabular} &
  \begin{tabular}[c]{@{}c@{}}0.225\\ (0.034)\end{tabular} &
  \begin{tabular}[c]{@{}c@{}}0.243\\ (0.082)\end{tabular} &
  \begin{tabular}[c]{@{}c@{}}0.723\\ (0.020)\end{tabular} &
  \begin{tabular}[c]{@{}c@{}}0.505\\ (0.207)\end{tabular} \\
$0.5$ &
  \begin{tabular}[c]{@{}c@{}}0.102\\ (0.010)\end{tabular} &
  \begin{tabular}[c]{@{}c@{}}0.109\\ (0.010)\end{tabular} &
  \begin{tabular}[c]{@{}c@{}}0.391\\ (0.001)\end{tabular} &
  \begin{tabular}[c]{@{}c@{}}0.216\\ (0.042)\end{tabular} &
  \begin{tabular}[c]{@{}c@{}}0.241\\ (0.081)\end{tabular} &
  \begin{tabular}[c]{@{}c@{}}0.554\\ (0.012)\end{tabular} &
  \begin{tabular}[c]{@{}c@{}}0.200\\ (0.076)\end{tabular} \\
$0.8$ &
  \begin{tabular}[c]{@{}c@{}}0.080\\ (0.007)\end{tabular} &
  \begin{tabular}[c]{@{}c@{}}0.101\\ (0.009)\end{tabular} &
  \begin{tabular}[c]{@{}c@{}}0.163\\ (0.002)\end{tabular} &
  \begin{tabular}[c]{@{}c@{}}0.296\\ (0.012)\end{tabular} &
  \begin{tabular}[c]{@{}c@{}}0.329\\ (0.139)\end{tabular} &
  \begin{tabular}[c]{@{}c@{}}0.399\\ (0.144)\end{tabular} &
  \begin{tabular}[c]{@{}c@{}}0.126\\ (0.014)\end{tabular} \\
$1$ &
  \begin{tabular}[c]{@{}c@{}}0.051\\ (0.012)\end{tabular} &
  \begin{tabular}[c]{@{}c@{}}0.010\\ (0.010)\end{tabular} &
  \begin{tabular}[c]{@{}c@{}}0.047\\ (0.011)\end{tabular} &
  \begin{tabular}[c]{@{}c@{}}0.433\\ (0.009)\end{tabular} &
  \begin{tabular}[c]{@{}c@{}}0.406\\ (0.143)\end{tabular} &
  \begin{tabular}[c]{@{}c@{}}0.342\\ (0.219)\end{tabular} &
  \begin{tabular}[c]{@{}c@{}}0.139\\ (0.022)\end{tabular} \\ \hline
\end{tabular}
\caption{RMISE as a point estimation accuracy metric for all methods under five true information-borrowing ratios. The values in the grid cells are the average over 10 independently generated datasets, and the standard deviation is shown in the brackets.}
\label{tab:pe-rmise}
\end{table}

\begin{table}[]
\centering
\begin{tabular}{ccccccc}
\hline
 &
  \multicolumn{3}{c}{MCMC} &
  \multicolumn{3}{c}{SVI} \\
 $\varepsilon_{\text{true}}$ & 
  RANDOM &
  IDIO &
  COMMON &
  RANDOM &
  IDIO &
  COMMON \\ \hline
$0$ &
  \begin{tabular}[c]{@{}c@{}}0.728\\ (0.096)\end{tabular} &
  \begin{tabular}[c]{@{}c@{}}0.776\\ (0.089)\end{tabular} &
  \begin{tabular}[c]{@{}c@{}}0.035\\ (0.003)\end{tabular} &
  \begin{tabular}[c]{@{}c@{}}0.138\\ (0.106)\end{tabular} &
  \begin{tabular}[c]{@{}c@{}}0.351\\ (0.189)\end{tabular} &
  \begin{tabular}[c]{@{}c@{}}0.013\\ (0.003)\end{tabular} \\
$0.2$ &
  \begin{tabular}[c]{@{}c@{}}0.796\\ (0.089)\end{tabular} &
  \begin{tabular}[c]{@{}c@{}}0.843\\ (0.116)\end{tabular} &
  \begin{tabular}[c]{@{}c@{}}0.045\\ (0.005)\end{tabular} &
  \begin{tabular}[c]{@{}c@{}}0.059\\ (0.017)\end{tabular} &
  \begin{tabular}[c]{@{}c@{}}0.263\\ (0.084)\end{tabular} &
  \begin{tabular}[c]{@{}c@{}}0.018\\ (0.004)\end{tabular} \\
$0.5$ &
  \begin{tabular}[c]{@{}c@{}}0.822\\ (0.096)\end{tabular} &
  \begin{tabular}[c]{@{}c@{}}0.902\\ (0.079)\end{tabular} &
  \begin{tabular}[c]{@{}c@{}}0.071\\ (0.008)\end{tabular} &
  \begin{tabular}[c]{@{}c@{}}0.048\\ (0.022)\end{tabular} &
  \begin{tabular}[c]{@{}c@{}}0.244\\ (0.101)\end{tabular} &
  \begin{tabular}[c]{@{}c@{}}0.029\\ (0.009)\end{tabular} \\
$0.8$ &
  \begin{tabular}[c]{@{}c@{}}0.829\\ (0.049)\end{tabular} &
  \begin{tabular}[c]{@{}c@{}}0.870\\ (0.046)\end{tabular} &
  \begin{tabular}[c]{@{}c@{}}0.176\\ (0.037)\end{tabular} &
  \begin{tabular}[c]{@{}c@{}}0.052\\ (0.009)\end{tabular} &
  \begin{tabular}[c]{@{}c@{}}0.244\\ (0.148)\end{tabular} &
  \begin{tabular}[c]{@{}c@{}}0.036\\ (0.017)\end{tabular} \\
$1$ &
  \begin{tabular}[c]{@{}c@{}}0.801\\ (0.104)\end{tabular} &
  \begin{tabular}[c]{@{}c@{}}0.896\\ (0.100)\end{tabular} &
  \begin{tabular}[c]{@{}c@{}}0.801\\ (0.150)\end{tabular} &
  \begin{tabular}[c]{@{}c@{}}0.450\\ (0.319)\end{tabular} &
  \begin{tabular}[c]{@{}c@{}}0.193\\ (0.171)\end{tabular} &
  \begin{tabular}[c]{@{}c@{}}0.190\\ (0.347)\end{tabular} \\ \hline
\end{tabular}
\caption{Coverage rate as an uncertainty estimation accuracy metric for all methods under five true information-borrowing ratios. The values in the grid cells are the average over 10 independently generated datasets, and the standard deviation is shown in the brackets.}
\label{tab:ue-acr}
\end{table}

\begin{table}[]
\centering
\begin{tabular}{ccccccc}
\hline
 &
  \multicolumn{3}{c}{MCMC} &
  \multicolumn{3}{c}{SVI} \\
 $\varepsilon_{\text{true}}$
 &
  RANDOM &
  IDIO &
  COMMON &
  RANDOM &
  IDIO &
  COMMON \\ \hline
$0$ &
  \begin{tabular}[c]{@{}c@{}}0.014\\ (0.005)\end{tabular} &
  \begin{tabular}[c]{@{}c@{}}0.014\\ (0.008)\end{tabular} &
  \begin{tabular}[c]{@{}c@{}}0.629\\ (0.003)\end{tabular} &
  \begin{tabular}[c]{@{}c@{}}0.255\\ (0.002)\end{tabular} &
  \begin{tabular}[c]{@{}c@{}}0.118\\ (0.098)\end{tabular} &
  \begin{tabular}[c]{@{}c@{}}0.719\\ (0.005)\end{tabular} \\
$0.2$ &
  \begin{tabular}[c]{@{}c@{}}0.014\\ (0.005)\end{tabular} &
  \begin{tabular}[c]{@{}c@{}}0.015\\ (0.005)\end{tabular} &
  \begin{tabular}[c]{@{}c@{}}0.492\\ (0.005)\end{tabular} &
  \begin{tabular}[c]{@{}c@{}}0.293\\ (0.030)\end{tabular} &
  \begin{tabular}[c]{@{}c@{}}0.132\\ (0.066)\end{tabular} &
  \begin{tabular}[c]{@{}c@{}}0.600\\ (0.015)\end{tabular} \\
$0.5$ &
  \begin{tabular}[c]{@{}c@{}}0.014\\ (0.005)\end{tabular} &
  \begin{tabular}[c]{@{}c@{}}0.014\\ (0.004)\end{tabular} &
  \begin{tabular}[c]{@{}c@{}}0.295\\ (0.005)\end{tabular} &
  \begin{tabular}[c]{@{}c@{}}0.388\\ (0.044)\end{tabular} &
  \begin{tabular}[c]{@{}c@{}}0.132\\ (0.050)\end{tabular} &
  \begin{tabular}[c]{@{}c@{}}0.448\\ (0.011)\end{tabular} \\
$0.8$ &
  \begin{tabular}[c]{@{}c@{}}0.011\\ (0.002)\end{tabular} &
  \begin{tabular}[c]{@{}c@{}}0.012\\ (0.002)\end{tabular} &
  \begin{tabular}[c]{@{}c@{}}0.100\\ (0.002)\end{tabular} &
  \begin{tabular}[c]{@{}c@{}}0.150\\ (0.009)\end{tabular} &
  \begin{tabular}[c]{@{}c@{}}0.188\\ (0.088)\end{tabular} &
  \begin{tabular}[c]{@{}c@{}}0.331\\ (0.128)\end{tabular} \\
$1$ &
  \begin{tabular}[c]{@{}c@{}}0.006\\ (0.003)\end{tabular} &
  \begin{tabular}[c]{@{}c@{}}0.011\\ (0.004)\end{tabular} &
  \begin{tabular}[c]{@{}c@{}}0.005\\ (0.002)\end{tabular} &
  \begin{tabular}[c]{@{}c@{}}0.008\\ (0.005)\end{tabular} &
  \begin{tabular}[c]{@{}c@{}}0.218\\ (0.067)\end{tabular} &
  \begin{tabular}[c]{@{}c@{}}0.276\\ (0.188)\end{tabular} \\ \hline
\end{tabular}
\caption{Interval score as an uncertainty estimation accuracy metric for all methods under five true information-borrowing ratios. The values in the grid cells are the average over 10 independently generated datasets, and the standard deviation is shown in the brackets.}
\label{tab:ue-is}
\end{table}

\section{Application: Modeling order flow in financial markets}\label{se:illustration}

%In the past two decades, the proliferation of electronic trading systems has led to financial instruments such as stocks and futures being more frequently traded in order-driven markets. In order-driven markets, buy and sell orders are submitted by traders to an electronic platform, which matches them with the best available offers \citep{gould2013limit}. Due to the nature of electronic trading, the arrival of orders tend to have high-frequency (usually within milliseconds). The order flow in such markets is of particular research interest, as it not only provides insights to high-frequency trading and order execution strategies \citep{alfonsi2010optimal}, but also aids in understanding the interplay of supply and demand and their roles in price formation \citep{cont2011statistical}.

{In this Section, we apply our method to study the order flow in Amazon's limit order book (LOB). We start by giving a brief introduction of the LOB data structure (we refer the readers to \citealp{gould2013limit} for a detailed description of LOB data).} LOB records the placement of limit orders from both buyers and sellers, with the submission time, proposed price and volume. %Note that a buy limit order of price $p$ and volume $\nu$ is a commitment to sell up to $\nu$ units of asset at a price no less that $p$, and vice versa. 
Once a limit order is placed, the order-matching algorithm of the platform attempts to match it with a pre-existing order of the other trade direction. A successful matching is called \textit{market order}. If the placed order does not match with any orders, it remains an \textit{active order} and is recorded by LOB until it is matched or \textit{cancelled}. Within the scope of this paper, we will look at the arrival of three event types: placement of active limit orders (submissions), market orders, and order cancellations. 

Our LOB data comes from LOBSTER (Limit Order Book System - The Efficient Reconstructor, \citep{huang2011lobster}), which is based on the official NASDAQ Historical TotalView-ITCH sample. The dataset includes the order book for events between 9:30:00 AM and 16:00:00 PM on June 21, 2012. We focus on modeling the order flow of level 1 data, that is, the order with the best \textit{bid}(\textit{ask}) prices, defined as the highest (lowest) price at which there is an active limit order. We exclude orders with trading volume lower than 100. The dataset includes 30411 order events with the timestamps (with nanosecond decimal precision). We group the events into four dimensions: buy submissions, buy market orders/cancellations, sell submissions, and sell market orders/cancellations, each containing 8409, 6791, 8801 and 6410 events, respectively. The flow of order events can be seen as a four-dimensional counting process, which can be modelled using an MHP.

To motivate the use of nonparametric estimators for the excitation functions of MHPs in this context, we present in Figure \ref{fig:motiv} the density plots of the cross-dimensional interarrival times among all four dimensions. Note that we excluded the interarrival times over 1 second as events are very unlikely to trigger events in the far future. There are two major implications of the figure. First, the density plots for interarrival times demonstrate various degrees of heterogeneity across dimensions, while having a similar shape. Another discovery is that the majority of the plots are multimodal, with a small `bump' roughly around the 0.5 second mark; one hypothesis is that this bump is the result of automated algorithms responding to market events.

We fit our model to the dataset using both MCMC and SVI methods. For both methods, we let $T_0 = 1$. We collected 10,000 posterior samples from the MCMC algorithms, with the first 5,000 discarded as burn-in. To prevent convergence to  sub-optimal modes, we fit the model to the same dataset under 50 different starting values with different random seeds, and keep the instance with the highest posterior mean marginal likelihood. Figure \ref{fig:app-mcmc} shows the posterior mean for the pointwise triggering kernels. The graphs show different decay rates across dimensions. Remarkably, there are two small bumps for transitions: from `buy submissions' to `buy market order/cancellations' and from `sell submissions' to `sell market order/cancellations'. Similarly, we run the SVI algorithm to fit the dataset with 100 different instances and choose the one with the highest ELBO to evaluate the parameter estimation results. We generated 5,000 samples from the variational distribution obtained from the converged model. Figure \ref{fig:app-svi} (shown in supplemental material section E) shows the variational mean for the pointwise triggering kernels. Similar to Figure \ref{fig:app-mcmc}, the triggering kernels decay differently, and there is a clear bump from `sell submissions' to `sell market order/cancellations'.

We also compared the estimation results for the $\boldsymbol{\alpha}$ matrix across both methods. Figure \ref{fig:alpha_heat} shows the heatmap for the point estimates (upper panel), $95\%$ credible interval lengths (lower panel) and Figure \ref{fig:alpha_spec} shows histograms of the spectral radius of the $\boldsymbol{\alpha}$ matrix for the MCMC (left panel) and SVI (right panel) algorithms. The spectral radius $\rho(\boldsymbol{\alpha})$ is defined as the absolute value of the largest eigenvalue of the $\boldsymbol{\alpha}$ matrix:
\begin{align*}
    \rho(\boldsymbol{\alpha})= \max_{\lambda \in \mathcal{E}(\boldsymbol{\alpha})} |\lambda|,
\end{align*}
where $\lambda \in \mathcal{E}(\boldsymbol{\alpha})$ denotes the set of eigenvalues of $\boldsymbol{\alpha}$. One sufficient condition for the stationarity of the MHP is that $\rho(\boldsymbol{\alpha}) < 1$ \citep{Hawkes1971, abergel2015long}. Hence, from Figure \ref{fig:alpha_spec}, we see that both MCMC and SVI suggest the process we study here is stationary.  However, the posterior distribution for the spectral radius under SVI suggest a higher value, and somewhat surprisingly, also a higher variability. On the other hand, Figure \ref{fig:alpha_heat} suggests that there are stronger interactions within the `buy' or `sell' transactions. This result is in line with the empirical results in Figure \ref{fig:motiv}, where the same structure can be seen.
\begin{figure}
    \centering
    \includegraphics[width = 0.9\linewidth]{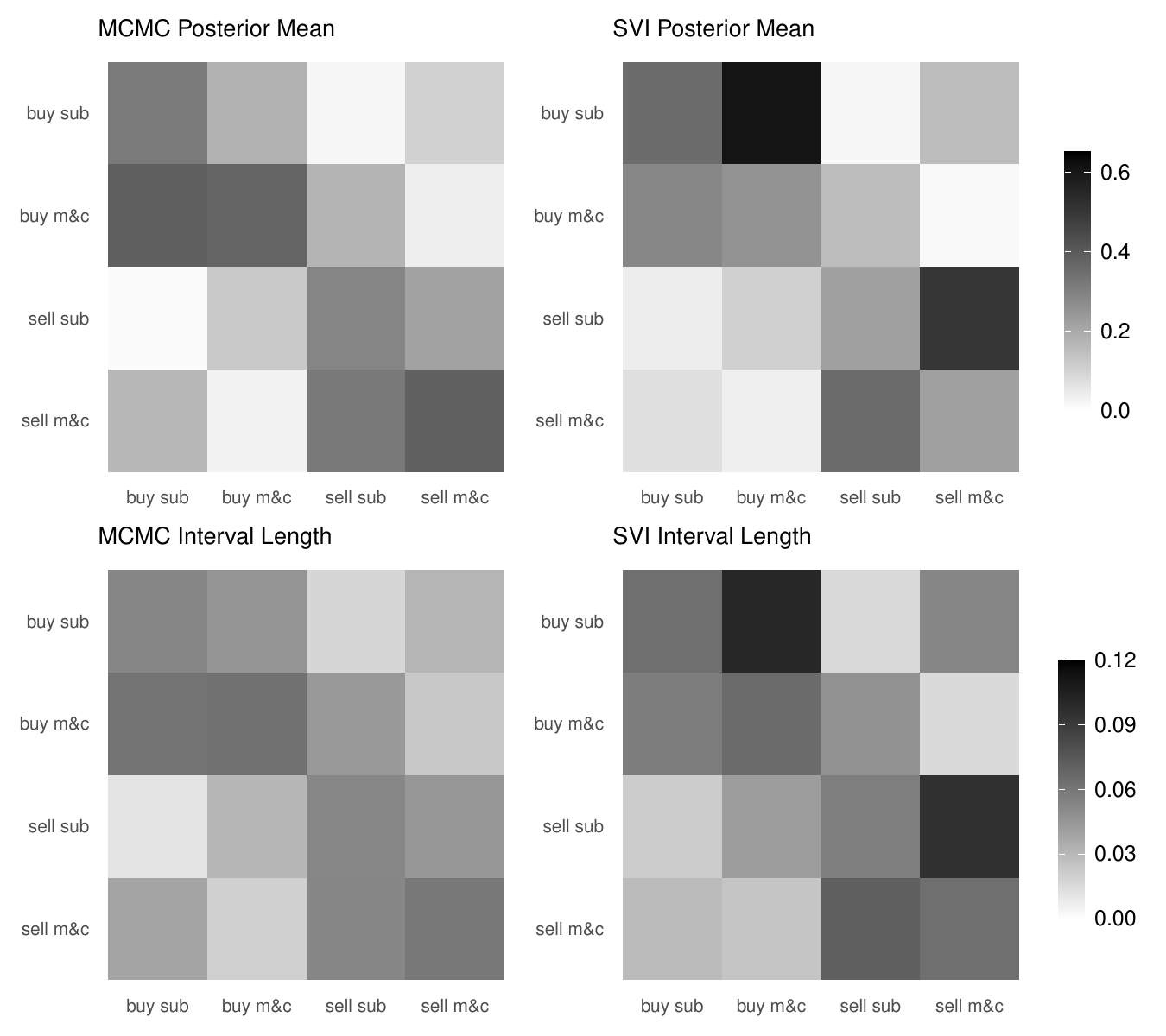}
    \caption{Estimation results for the $\boldsymbol{\alpha}$ matrix. The heatmaps show the posterior mean (first row) and average $95\%$ posterior interval lengths (second row) for the values in $\boldsymbol{\alpha}$ by MCMC (left panel) and SVI (right panel) algorithms, for buy submissions (`buy sub'), buy market order/cancellations (`buy m\&c'), sell submissions (`sell sub'), sell market order/cancellations (`sell m\&c'). y-axis shows the parent events, while x-axis shows the child events. Darker colors corresponds to higher values.}
    \label{fig:alpha_heat}
\end{figure}
\begin{figure}
    \centering
    \includegraphics[width = 0.9\linewidth]{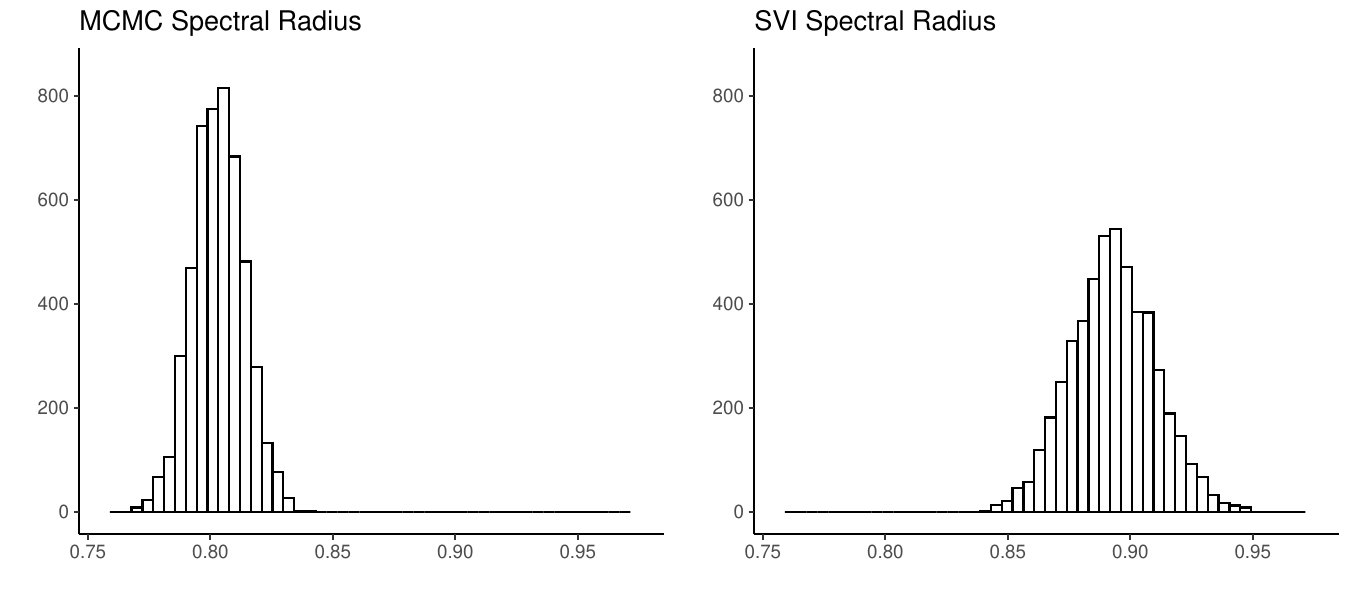}
    \caption{Histograms for the spectral radius of the $\boldsymbol{\alpha}$ matrix by the MCMC (left) and SVI (right) algorithm.}
    \label{fig:alpha_spec}
\end{figure}

\begin{figure}
    \centering
    \includegraphics[width=\linewidth]{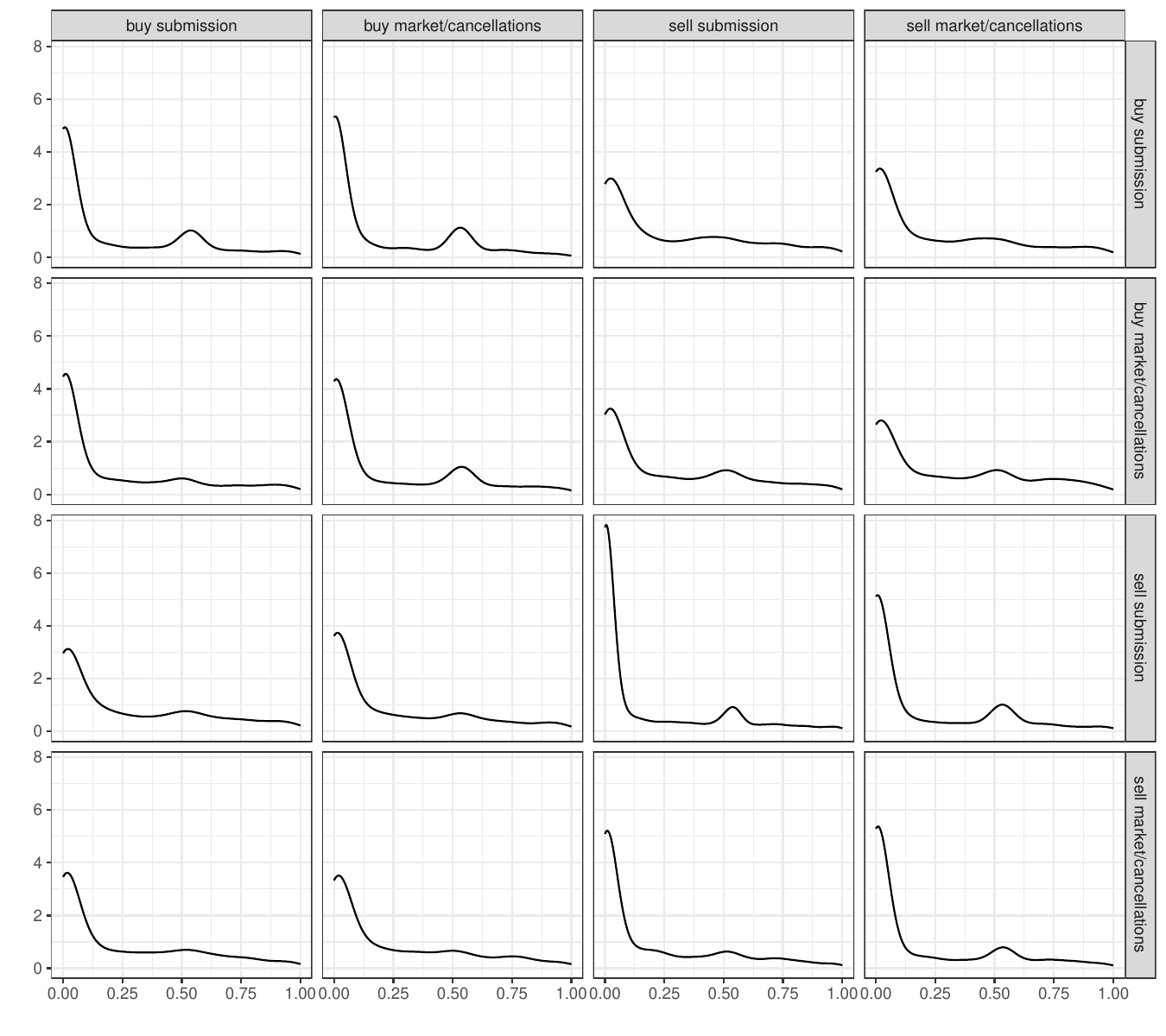}
    \caption{Empirical density plots of the cross-dimensional inter-arrival times (within one second) among the four dimensions.}
    \label{fig:motiv}
\end{figure}

\begin{figure}
    \centering
    \includegraphics[width=\linewidth]{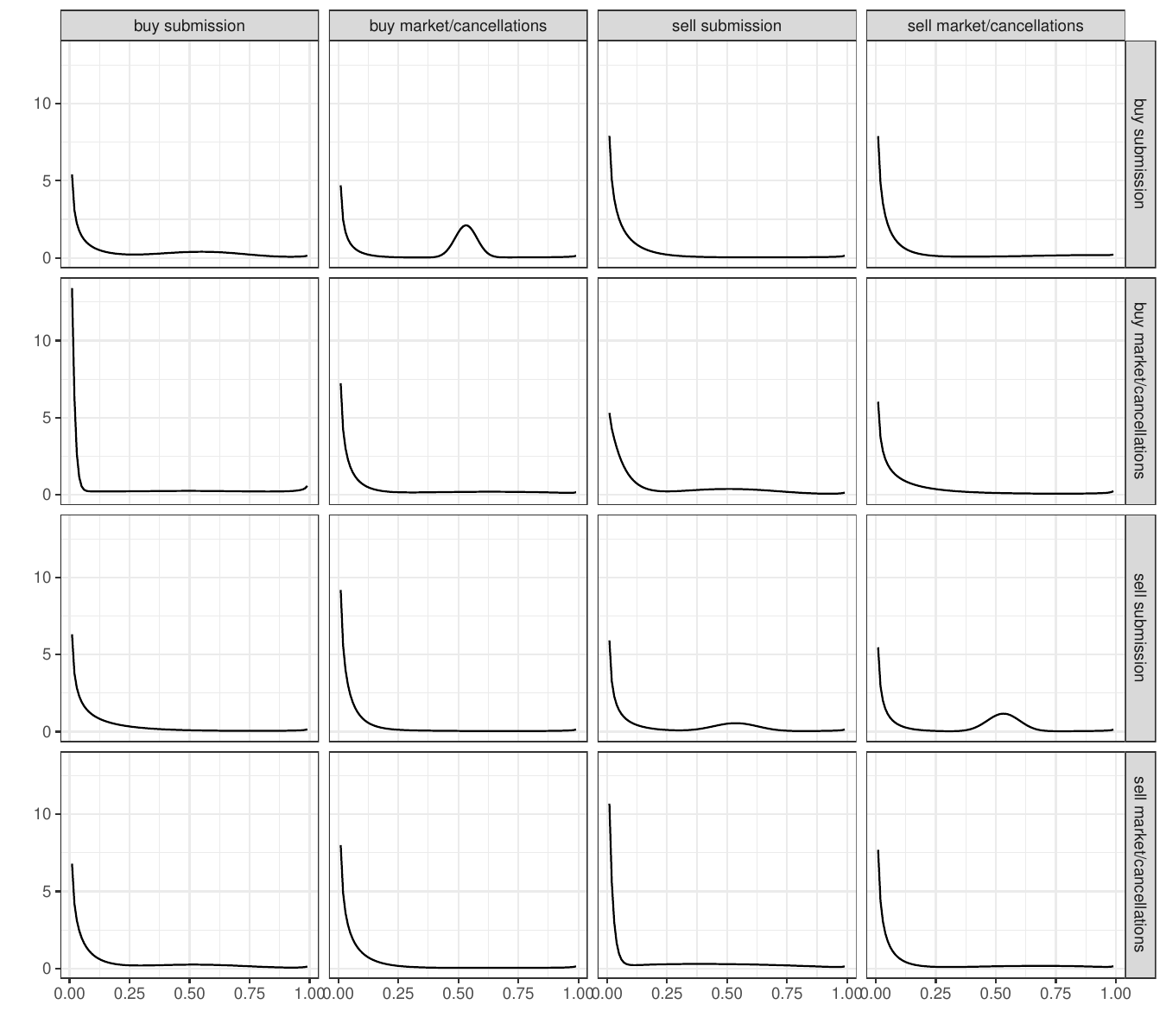}
    \caption{Estimated density plots from the MCMC algorithm of the cross-dimensional inter-arrival times (within one second) among the four dimensions.}
    \label{fig:app-mcmc}
\end{figure}

%\begin{itemize}
%    \item Identity the scientific question (look at previous literature on this), and say LOB is the data to look at
%    \item Background of the dataset (limit order book dynamics)
%    \item Describe the dataset, how we filter the dataset by trade volumes, identify the dimensions
%    \item Descriptive analysis about the data as motivation: plots on inter-arrival empirical distributions 
%    \item Results from two methods
%\end{itemize}

\section{Discussion}

This paper proposed a semiparametric model for the excitation functions of a multivariate Hawkes process that relies on a dependent Dirichlet Process mixture of scaled Beta functions. We showed that the model has large Kulback-Leibler support and developed two computation methods for the MHP-DDP model: a Metropolis-within-Gibbs Markov chain Monte Carlo sampler, and a stochastic variational algorithm that is scalable to large-scale datasets. 

Our simulation examples in Section \ref{se:sim} showed that our model outperforms (in both point and uncertainty estimation) the benchmark methods when the true excitation functions are mixtures of common and idiosyncratic components. This suggests that our estimation method could be particularly useful when the triggering patterns across dimensions are different, but common information could also be borrowed. This is evidenced by our application example in Section \ref{se:illustration}, where the empirical distribution of inter-arrival times for order book events demonstrate dimensional-specific heterogeneities, but also share common patterns.

There are a few directions where our model could be extended. In addition to the dependent Dirichlet prior construction we employed based on \cite{Muller2004}, it is also of interest to explore the dependent structure in hierarchical Dirichlet processes \citep{teh2004sharing} or nested Dirichlet processes \cite{rodriguez2008nested}. Secondly, the model could be extended to the case of nonlinear Hawkes processes \citep{bremaud1996stability}. The case where we model $\log \lambda_k(t)$ instead of $\lambda_k(t)$ as the sum of the background intensity and the triggering kernels is of particular interest, as past literature has demonstrated the need for modeling flexible and multi-modal triggering kernel patterns (for example, see \citealp{rambaldi2017role}). 

%\appendix

\bibliographystyle{Chicago}
\bibliography{Bibliography-MM-MC}

\newpage 
\bigskip
\begin{center}
{\large\bf SUPPLEMENTARY MATERIAL}
\end{center}
\appendix
\counterwithin{figure}{section}
\counterwithin{table}{section}

\setcounter{page}{1}

\section{Proof of Theorem \ref{th:support}}\label{ap:KLsupport}

\begin{proof}
We show the KL property of $f_0$ by proving Conditions A1-A3 from Theorem 1 in \cite{wu2008kullback} holds. Since we don't have $\phi$ in our case, Condition A2 is automatically satisfied. The remaining proof is similar to Theorem 11 in \citep{wu2008kullback}, except that the mixing distribution is drawn from a Dirichlet process, and that the two location parameters are also being mixed. For $\forall f_0 \in \mathcal{H}_0$ and $\forall \varepsilon > 0$, there exists a finite Beta mixture function $f_{P_{\varepsilon}}$ with ${H}$ mixtures where
%\begin{align*}
%    \left|f_0(x) - f_{P_{\varepsilon}}(x)\right| < e^{\varepsilon}, \forall x \in %[0,1],
%\end{align*}
%where
\begin{align}
\label{eq:bst}
    f_{P_{\varepsilon}}(x) &= \sum_{k=1}^{H} \frac{f_0\left(\frac{k-1}{H-1}\right)}{\sum_{k=1}^H f_0\left(\frac{k-1}{H-1}\right)}f_{\text{Beta}}(x \mid k, H-k) = \int f_{\text{Beta}}(x \mid a, b) dP_{\varepsilon},
\end{align}
and 
\begin{align}
\label{eq:pars}
P_{\varepsilon} = \sum_{k=1}^H \omega_k \delta_{a,b}(a_k^0, b_k^0),\quad \omega_k^0 = \frac{f_0\left(\frac{k-1}{H-1}\right)}{\sum_{k=1}^H f_0\left(\frac{k-1}{H-1}\right)},\quad a_k^0 = k, \quad b_k^0 = H-k,
\end{align}
such that 
\begin{align*}
    \int_0^1 f_0(x)\log \frac{f_0(x)}{f_{P_\varepsilon}(x)}dx < \varepsilon.
\end{align*}
Thus Condition A1 holds. We then show Condition A3 also holds. First, we define $\mathcal{C}_{\omega} \subset \mathbb{S}^{H-1}, \mathcal{C}^{h}_{ab} \subset \mathbb{R}^{2}, h = 1, \dots , H$ as sets such that 
\begin{align*}
    \mathcal{C}_{\omega} &= \left\{(\omega_1, \dots, \omega_H) : \omega_h > \omega_h^0 e^{-\frac{\varepsilon}{4}}, ~~ \sum_{h=1}^H \omega_h = 1\right\},
\end{align*}
\begin{equation}
\label{eq:eq3}
\begin{aligned}
     \mathcal{C}_{ab}^h &= \left\{(a_h, b_h): a_h  < a_h^0, b_h < b_h^0,\right.\\ &\left.(a_h^0 - a_h)\left(\log x_M + \psi(a_h + b_h) - \psi(a_h) \right) + (b_h^0 - b_h)\left(\log (1-x_M) + \psi(b_h + a_h) - \psi(b_h) \right) < \frac{\varepsilon}{4}\right\},   
\end{aligned}
\end{equation}
where 
\begin{align*}
    x_M = \max\left\{d, 1-d, \frac{a_h^0 - a_h}{a_h^0 - a_h + b_h^0 - b_h} \right\}.
\end{align*}
Finally, we let $\mathcal{C}_{ab} = \oplus_{h=1}^H \mathcal{C}_{ab}^h$, $\boldsymbol{\omega} := (\omega_1, \dots, \omega_H) \in \mathbb{S}^{H-1}, \boldsymbol{a} := (a_1, \dots, a_H) \in \mathbb{R}^H, \boldsymbol{b} := (b_1, \dots, b_H) \in \mathbb{R}^H$ and let $\mathcal{W} \subset \mathcal{H}$ be the set of finite mixture distributions induced by $\mathcal{C}_{\omega}$ and $\mathcal{C}_{ab}$:
\begin{align*}
    \mathcal{W} := \left\{P \in \mathcal{W} \mid P =\sum_{k=1}^H \omega_k \delta_{a, b}\left(a_k, b_k\right), \boldsymbol{\omega} \in \mathcal{C}_{\omega}, (a_h, b_h) \in \mathcal{C}_{ab}^h \right\}.
\end{align*}
We note that for the chosen $\boldsymbol{\omega}_0, \boldsymbol{a}_0, \boldsymbol{b}_0$, for any $\omega \in \mathcal{C}_{\omega}$, we have 
\begin{align}
\label{eq:omega}
    \frac{\sum_{j=1}^H \omega^0_h f_{\text {Beta }}\left(x \mid a_h^0, b_h^0\right)}{\sum_{j=1}^H \omega_h f_{\text {Beta }}\left(x \mid a_h^0, b_h^0\right)}<e^{\frac{\varepsilon}{4}}, \quad \forall x \in(d, 1-d).
\end{align}
We then note that 
\begin{align}
\label{eq:x_M}
    (a_h^0 - a_h)\log x + (b_h^0 - b_h)\log (1-x) \leq (a_h^0 - a_h)\log x_M + (b_h^0 - b_h)\log (1-x_M), ~~ \forall x \in (d, 1-d).
\end{align}
Thus for any $(a_h, b_h) \in \mathcal{C}_{ab}^h$, consider the Cauchy remainder form of the first-order expansion for $l(x, a_h^0, b_h^0):=\log f_{\text {Beta }}(x ; a_h^0, b_h^0)$ at $(a_h, b_h)$, we have
\begin{equation}
\begin{aligned}
\label{eq:varab}
    l(x&; a_h^0, b_h^0) = l(x; a_h, b_h) + \nabla l(x; a_h, b_h)^T \begin{bmatrix}a_h^0 - a_h \\ b_h^0 - b_h\end{bmatrix}  +\frac{\theta^2}{2}\begin{bmatrix}a_h^0 - a_h & b_h^0 - b_h\end{bmatrix}\nabla^2 l(x; a_h^*, b_h^*) \begin{bmatrix}a_h^0 - a_h \\ b_h^0 - b_h\end{bmatrix} \\
    &< l(x; a_h, b_h) + (a_h^0 - a_h)\left(\log x + \psi(a_h + b_h) - \psi(a_h) \right) + (b_h^0 - b_h)\left(\log (1-x) + \psi(b_h + a_h) - \psi(b_h) \right) \\
    &< l(x; a_h, b_h) + (a_h^0 - a_h)\log x_M + (b_h^0 - b_h)\log (1-x_M) \\
    &< l(x; a_h, b_h) + \frac{\varepsilon}{4} , ~~ \forall x \in (d, 1-d).
\end{aligned}
\end{equation}
Note that the chain of inequalities are based on equations \eqref{eq:eq3} and \eqref{eq:x_M}, and the fact that 
$$\nabla^2 l(x ; a, b) \prec 0, a > 0, b > 0.$$
Finally, note that \eqref{eq:varab} is equivalent to 
\begin{align}
\label{eq:varab2}
\frac{f_{\text {Beta }}\left(x \mid a_h^0, b_h^0\right)}{f_{\text {Beta }}\left(x \mid a_h, b_h\right)}<e^{\frac{\varepsilon}{4}}, \quad \forall x \in(d, 1-d),
\end{align}
and if for $h = 1, \dots, H$, $(a_h, b_h) \in \mathcal{C}_{ab}^h$, for any $(\omega_1, \dots, \omega_H) \in \mathcal{C}_{\omega}$, we have 
$$\frac{\sum_{j=1}^H \omega_h f_{\text {Beta }}\left(x \mid a_h^0, b_h^0\right)}{\sum_{j=1}^H \omega_h f_{\text {Beta }}\left(x \mid a_h, b_h\right)}<e^{\frac{\varepsilon}{4}}, \quad \forall x \in(d, 1-d).$$
Combining equations \eqref{eq:omega} and \eqref{eq:varab2}, we have, for any $\boldsymbol{\omega} \in \mathcal{C}_{\omega}, (a_h, b_h) \in \mathcal{C}_{ab}^h, h = 1, \dots, H$, we have 
\begin{align*}
    \frac{\sum_{j=1}^H \omega_h^0 f_{\text {Beta }}\left(x \mid a_h^0, b_h^0\right)}{\sum_{j=1}^H \omega_h f_{\text {Beta }}\left(x \mid a_h, b_h\right)} = \frac{\sum_{j=1}^H \omega_h^0 f_{\text {Beta }}\left(x \mid a_h^0, b_h^0\right)}{\sum_{j=1}^H \omega_h f_{\text {Beta }}\left(x \mid a_h^0, b_h^0\right)} \cdot \frac{\sum_{j=1}^H \omega_h f_{\text {Beta }}\left(x \mid a_h^0, b_h^0\right)}{\sum_{j=1}^H \omega_h f_{\text {Beta }}\left(x \mid a_h, b_h\right)} = e^{\frac{\varepsilon}{2}}, \quad \forall x \in (d, 1-d).
\end{align*}
Thus 
\begin{align}
\label{eq:int1}
\int_d^{1-d} f_0(x) \log \frac{f_{P_{\varepsilon}}(x)}{f_{P}(x)}dx < \int_d^{1-d} f_0(x) dx \cdot \frac{\varepsilon}{2} < \frac{\varepsilon}{2}.
\end{align}
We then consider the scenario where $x \in (0, d] \bigcup [1-d, 1)$. We want to show that the likelihood ratio $ \frac{f_{\text {Beta }}\left(x \mid a_h^0, b_h^0\right)}{f_{\text {Beta }}\left(x \mid a_h, b_h\right)}$ has a uniform finite upper bound over all $x \in (0, d] \bigcup [1-d, 1)$ and $(a_h, b_h) \in \mathcal{C}_{ab}^h, h = 1,\dots,H$, i.e.
\begin{align*}
    \sup_{x \in (0, d] \bigcup [1-d, 1)}\frac{f_{\text {Beta }}\left(x \mid a_h^0, b_h^0\right)}{f_{\text {Beta }}\left(x \mid a_h, b_h\right)} &= \frac{\operatorname{Be}(a_h, b_h)}{\operatorname{Be}(a_h^0, b_h^0)}\sup_{x \in (0, d] \bigcup [1-d, 1)}x^{a_h^0 - a_h}(1-x)^{b_h^0 - b_h} \\&= \frac{\operatorname{Be}(a_h, b_h)}{\operatorname{Be}(a_h^0, b_h^0)} d^{a_h^0+b_h^0-a_h-b_h}\leq \frac{\operatorname{Be}(a_h, b_h)}{\operatorname{Be}(a_h^0, b_h^0)},
\end{align*}
which follows from the fact that $a_h < a_h^0, b_h < b_h^0$. Thus we have 
\begin{align*}
    \sup_{\substack{\boldsymbol{a}, \boldsymbol{b} \in \mathcal{C}_{ab} \\ \boldsymbol{\omega} \in \mathcal{C}_{\omega}}} \sup_{x \in (0, d] \bigcup [1-d, 1)} \frac{\sum_{j=1}^H \omega_h^0 f_{\text {Beta }}\left(x \mid a_h^0, b_h^0\right)}{\sum_{j=1}^H \omega_h f_{\text {Beta }}\left(x \mid a_h, b_h\right)} &\leq e^{-\frac{\varepsilon}{4}} \sup_{\substack{\boldsymbol{a}, \boldsymbol{b} \in \mathcal{C}_{ab}}} \frac{\operatorname{Be}\left(a_h, b_h\right)}{\operatorname{Be}\left(a_h^0, b_h^0\right)} 
    \leq  \sup_{\substack{\boldsymbol{a}, \boldsymbol{b} \in \mathcal{C}_{ab}}} \frac{\operatorname{Be}\left(a_h, b_h\right)}{\operatorname{Be}\left(a_h^0, b_h^0\right)} := M <+\infty.
\end{align*}
Thus we have 
\begin{align*}
    \int_0^d f_0(x) \log \frac{f_{P_{\varepsilon}}(x)}{f_P(x)} d x+\int_{1-d}^1 f_0(x) \log \frac{f_{P_{\varepsilon}}(x)}{f_P(x)} d x<M\left(F_0(d)+1-F_0(1-d)\right).
\end{align*}
Thus, we can choose $d$ small enough such that $M\left(F_0(d)+1-F_0(1-d)\right) < \frac{\varepsilon}{2}$, such that 
\begin{align*}
    \int_0^1 f_0(x)\log \frac{f_{P_{\varepsilon}}(x)}{f_P(x)} d x &= \int_0^d f_0(x) \log \frac{f_{P_{\varepsilon}}(x)}{f_P(x)} d x+\int_{1-d}^1 f_0(x) \log \frac{f_{P_{\varepsilon}}(x)}{f_P(x)} d x + \int_{d}^{1-d} f_0(x) \\
    &< \frac{\varepsilon}{2} + \frac{\varepsilon}{4} + \frac{\varepsilon}{4} = \varepsilon.
\end{align*}
Finally, as $\mathcal{C}_{\omega}, \mathcal{C}_{ab}$ are nonempty and open sets, we have $\Pi(\mathcal{W}) > 0$, thus $f_0 \in KL(\Pi)$.
\end{proof}

\begin{corollary} Consider $K$ continuous densities on $[0,1]$, i.e. let $f^1_0, \dots, f^K_0 \in \mathcal{H}_0$. Consider the following joint prior $\Pi^* : \oplus_{h=1}^H \mathcal{H}_0 \rightarrow \mathbb{R}$ for $(f^1, \dots, f^K)$ such that 
\begin{equation}
\begin{aligned}
    f^k &= \delta g_0 + (1-\delta) g^k, k = 1, \dots, K, \\
    g_k(x) &= \int f_{\operatorname{Beta}}(x \mid a, b) dP(a,b), k = 0, 1, \dots, K, \\
    P &\sim \operatorname{DP}(\gamma, \operatorname{Gamma}(c_a, d_a)\times \operatorname{Gamma}(c_b, d_b)), \\
    \delta &\sim \operatorname{Dirichlet}(1,1),
\end{aligned}
\end{equation}
we have,
\begin{align*}
    \Pi^*\left(\{f^1, \dots, f^K: \operatorname{KL}(f_0^k, f^k) < \varepsilon, k = 1, \dots, K\}\right) > 0.
\end{align*}
\end{corollary}
\begin{proof}
Let $f_0^0(x) := \frac{1}{K} \sum_{k=1}^K f^{k}_0(x)$. For $\varepsilon > 0$, we let 
\begin{align*}
    \mathcal{G}_{0} = \{g^{0} \in \mathcal{H}_0 : \text{KL}(f_0^0, g^{0}) < \frac{\varepsilon}{2}\}, \quad
    \mathcal{G}^{k} = \{g^{k} \in \mathcal{H}_0 : \text{KL}(f^k_{0}, g^k) < \frac{\varepsilon}{2}\}.
\end{align*}
Note that from Theorem 1.1 we have $\Pi_{g^{k}}\left(g^k \in \mathcal{G}^k\right) > 0, k = 0, 1, \dots, K$.
Let $M := \max_{1 \leq j, k \leq M} \sup_{g^0 \in \mathcal{G}_0} \text{KL}(f_{ij}^0, g_0) < +\infty$, we have, based on the convexity of KL divergence:
\begin{align*}
    \text{KL}(f^k_{0}, \delta g^0 + (1-\delta) g^k) &\leq \delta \text{KL}(f^k_{0}, g^0) + (1-\delta)\text{KL}(f^k_{0}, g^k) < \delta M + (1-\delta) \frac{\varepsilon}{2} < \varepsilon, 
\end{align*}
for all $\forall \delta < \frac{\varepsilon}{2M - \varepsilon}$, $g^k \in \mathcal{G}^k, k = 0, 1, \dots, K$.
Thus
\begin{align*}
&\Pi^*\left(\left\{f^1, \ldots, f^K: \operatorname{KL}\left(f_0^k, f^k\right)<\varepsilon, k=1, \ldots, K\right\}\right) \\
=&\Pi^*\left(\left\{g^0, g^1, \ldots, g^K,\delta: \operatorname{KL}\left(f_0^k, \delta g^0+(1-\delta) g^k\right)<\varepsilon, k=1, \ldots, K\right\}\right) \\
>&\Pi^*\left(\left\{g^0, g^1, \ldots, g^K,\delta: g^k \in \mathcal{G}^{k}, k = 0, 1, \dots, K, \delta < \frac{\varepsilon}{2M-\varepsilon} 
\right\}\right)\\
=& 
    \prod_{0 \leq k \leq K}\Pi_{g^{k}}(g^{k} \in \mathcal{G}^{k})\Pi_{\delta}(\delta < \frac{\varepsilon}{2M- \varepsilon}) > 0.
\end{align*}
\end{proof}

\section{Details of the MCMC algorithm}\label{se:mcmc_details}

\subsection{Sampling \texorpdfstring{$\boldsymbol{\mu}$}{mu} and \texorpdfstring{{$\boldsymbol{\alpha}$}}{alpha}}
\label{subsec:samp-alphamu}
Given the branching structure, the full conditional posteriors for MHP parameters $\boldsymbol{\mu}$ and $\boldsymbol{\alpha}$ are conjugate to their Gamma priors:
\begin{align*}
    \mu_\ell \mid \mathbf{B}, \mathbf{X} &\sim \text{Gamma}(e_\ell + |I_{\ell}|, f_{\ell} + T), ~~ \ell = 1, \dots, K \\
    \alpha_{k,\ell} \mid \mathbf{B}, \mathbf{X} &\sim \text{Gamma}\left(g_{k,\ell} + |O_{k,\ell}|, h_{k,\ell} + \sum_{d_i=k} \tilde{\Phi}_{k, \ell}\left(T-t_i\right)\right), ~~ k, \ell = 1, \dots, K.
\end{align*}

\subsection{Sampling \texorpdfstring{$\boldsymbol{b}$}{B}}
\label{subsec:samp-branch}
Given the model parameters, we can update the branching structure using the following formula. For $j = 1, \dots, n, i = 1, \dots j$, we have:
\begin{align*}
\label{eq:p}
	p\left(\mathbf{B}_{j,j}=1,\mathbf{B}_{j,-j}=0 \mid \boldsymbol{\mu}, \boldsymbol{\alpha}, \boldsymbol{a}, \boldsymbol{b}, \boldsymbol{p},  \mathbf{X}\right) &= \frac{\mu_{d_j}}{\mu_{d_j} + \sum_{i=1}^{j-1} \alpha_{d_i,d_j}\tilde{\phi}_{d_i, d_j}(t_j - t_i)} \\
    p\left(\mathbf{B}_{j,i}=1,\mathbf{B}_{j,-i}=0 \mid \boldsymbol{\mu}, \boldsymbol{\alpha}, \boldsymbol{a}, \boldsymbol{b}, \boldsymbol{p},  \mathbf{X}\right) &= \frac{\alpha_{d_i,d_j}\tilde{\phi}_{d_i, d_j}(t_j - t_i)}{\mu_{d_j} + \sum_{i\neq j} \alpha_{d_i,d_j}\tilde{\phi}_{d_i, d_j}(t_j - t_i)}.
\end{align*}

\subsection{Sampling \texorpdfstring{$\mathbf{Z}$}{Z}}
\label{subsec:sample-Z}

Given $\boldsymbol{b}$ and model parameters $\boldsymbol{a}^{0}, \boldsymbol{a}^{k,\ell}, \boldsymbol{b}^{0}, \boldsymbol{b}^{k,\ell}$, $\boldsymbol{p}^0$, $\boldsymbol{p}^{k,\ell}$ and $\varepsilon$, we can generate posterior samples of the latent allocation variables sequentially, from the following discrete distribution. If $B_{i,j}=1$ and $(i,j) \in \mathcal{T}_{i,j}$, 
\begin{align*}
    &P(W_{i,j}=0 \mid  \mathbf{B}, \mathbf{X}, \boldsymbol{a}, \boldsymbol{b}, \boldsymbol{p},\varepsilon) = \frac{\sum_{h=1}^{L_0}  \varepsilon p_h^0 f_{\text { Beta }}\left(t_j - t_i \mid a_h^0, b_h^0\right)}{\sum_{h=1}^{L_0} \varepsilon p_h^0 f_{\text { Beta }}\left(t_j - t_i \mid a_h^0, b_h^0\right) + \sum_{h=1}^{L} (1-\varepsilon) p_h^{k, \ell} f_{\text { Beta }}\left(t_j - t_i \mid a_h^{k, \ell}, b_h^{k, \ell}\right)} \\
    &P(W_{i,j}=1 \mid  \mathbf{B}, \mathbf{X}, \boldsymbol{a}, \boldsymbol{b}, \boldsymbol{p},\varepsilon) = \frac{  \sum_{h=1}^{L}(1-\varepsilon) p_h^{k, \ell} f_{\text { Beta }}\left(t_j - t_i \mid a_h^{k, \ell}, b_h^{k, \ell}\right)}{\sum_{h=1}^{L_0} \varepsilon p_h^0 f_{\text { Beta }}\left(t_j - t_i \mid a_h^0, b_h^0\right) + \sum_{h=1}^{L} (1-\varepsilon)p_h^{k, \ell} f_{\text { Beta }}\left(t_j - t_i \mid a_h^{k, \ell}, b_h^{k, \ell}\right)} \\
    &P(Z_{i,j} = h \mid W_{i,j}=0, \mathbf{B}, \mathbf{X}, \boldsymbol{a}, \boldsymbol{b}, \boldsymbol{p},\varepsilon) = \frac{ p_h^0 f_{\text { Beta }}\left(t_j - t_i \mid a_h^0, b_h^0\right) }{\sum_{h=1}^{L_0}  p_h^0 f_{\text { Beta }}\left(t_j - t_i \mid a_h^0, b_h^0\right)}, \text{ for }h = 1, \dots, L_0 \\ &P(Z_{i,j} = h \mid W_{i,j}=1, \mathbf{B}, \mathbf{X}, \boldsymbol{a}, \boldsymbol{b}
    , \boldsymbol{p},\varepsilon) = \frac{p_h^{k, \ell} f_{\text { Beta }}\left(t_j - t_i \mid a_h^{k, \ell}, b_h^{k, \ell}\right)}{\sum_{h=1}^{L} p_h^{k, \ell} f_{\text { Beta }}\left(t_j - t_i \mid a_h^{k, \ell}, b_h^{k, \ell}\right)}, \text{ for }h = 1, \dots, L.
\end{align*}

%$$
%P(Z_{i,j} = h \mid \boldsymbol{b}, \mathbf{X}, \boldsymbol{a}, \boldsymbol{b}, \boldsymbol{p},\varepsilon) \propto \left\{\begin{array}{l}
%\varepsilon p_{h}^0 \operatorname{Beta}\left(a_{h}^0, b_{h}^0\right)  ~~~~~~~~~~~~~~~ h \in \{1, \dots, L_0 \} \\
%(1-\varepsilon) p_{h}^{k,\ell}\operatorname{Beta}\left(a_{h}^{k,\ell}, b_{h}^{k,\ell}\right) ~~ h \in \{L_0 + 1, \dots, L_0 + L\}
%\end{array}.\right.
%$$

\subsection{Sampling \texorpdfstring{$\boldsymbol{a}$}{a} and \texorpdfstring{$\boldsymbol{b}$}{b}}
\label{subsec:samp-ab}
Up to a normalizing constant and approximation of the compensator, the full conditional distribution for $\boldsymbol{a}^{0}, \boldsymbol{a}^{k,\ell}, \boldsymbol{b}^{0}$ and $\boldsymbol{b}^{k,\ell}$ can be written as
\begin{align*}
p(a_h^{k, \ell} \mid \mathbf{X}, \mathbf{B}, \mathbf{Z}, \cdot) &\propto \left[\frac{\Gamma\left(a_h^{k, \ell}+b_h^{k, \ell}\right)}{\Gamma\left(a_h^{k, \ell}\right) }\right]^{N_{h}^{k,\ell}} \prod_{(i,j) \in \mathcal{T}_{k,\ell}} \left[ (t_j - t_i)^{a_h^{k,\ell}-1} \right]^{B_{i,j} W_{i,j}I(Z_{i,j} = h)}   {\left(a_{h}^{k, \ell}\right)}^{c_a - 1} e^{-d_a a_{h}^{k, \ell}} \\
p(a_h^0 \mid \mathbf{X}, \mathbf{B}, \mathbf{Z}, \cdot) &\propto \left[\frac{\Gamma\left(a_h^0+b_h^0\right)}{\Gamma\left(a_h^0\right) }\right]^{N_{h}^{0}} \prod_{k=1}^K\prod_{\ell=1}^K  \prod_{(i,j) \in \mathcal{T}_{k,\ell}} \left[ (t_j - t_i)^{a_h^0-1} \right]^{B_{i,j} (1-W_{i,j})I(Z_{i,j} = h)}  {\left(a_{h}^0\right)}^{c_a - 1} e^{-d_a a_{h}^0}\\
p(b_h^{k, \ell} \mid \mathbf{X}, \mathbf{B}, \mathbf{Z}, \cdot) &\propto \left[\frac{\Gamma\left(a_h^{k, \ell}+b_h^{k, \ell}\right)}{\Gamma\left(b_h^{k, \ell}\right) }\right]^{N_{h}^{k,\ell}} \prod_{(i,j) \in \mathcal{T}_{k,\ell}} \left[ [T_b - (t_j - t_i])^{b_h^{k,\ell}-1} \right]^{B_{i,j} W_{i,j}I(Z_{i,j} = h)}   {\left(b_{h}^{k, \ell}\right)}^{c_b - 1} e^{-d_b b_{h}^{k, \ell}} \\
p(b_h^0 \mid \mathbf{X}, \mathbf{B}, \mathbf{Z}, \cdot) &\propto \left[\frac{\Gamma\left(a_h^0+b_h^0\right)}{\Gamma\left(b_h^0\right) }\right]^{N_{h}^{0}} \prod_{k=1}^K\prod_{\ell=1}^K  \prod_{(i,j) \in \mathcal{T}_{k,\ell}} \left[ [T_b - (t_j - t_i)]^{b_h^0-1} \right]^{B_{i,j} (1-W_{i,j})I(Z_{i,j} = h)}  {\left(b_{h}^0\right)}^{c_b - 1} e^{-d_b b_{h}^0},
\end{align*}
where 
\begin{align*}
    N_{h}^{k,\ell} = \sum_{(i, j) \in \mathcal{T}_{k, \ell}} B_{i, j} (1-W_{i, j}) I\left(Z_{i, j}=h\right), & N_{h}^{0} = \sum_{k=1}^K \sum_{\ell =1}^K\sum_{(i, j) \in \mathcal{T}_{k, \ell}} B_{i, j} W_{i, j} I\left(Z_{i, j}=h\right).
\end{align*}
%Note that with the augmentation of $\boldsymbol{b}, \mathbf{W}$ and $\mathbf{Z}$, $N_{h}^{0}$ and $N_{h}^{k,\ell}$ can be interpreted as the number of transitions from dimension $k$ to $\ell$ that belongs to the $h$-th common or idiosyncratic Beta mixture component.
 As the posterior distributions for $\boldsymbol{a}^{k,\ell}, \boldsymbol{a}^0, \boldsymbol{b}^{k,\ell}, \boldsymbol{b}^0$ are not conjugate to their priors, we update these parameters through Metropolis-within-Gibbs updates on the log scale.

\subsection{Sampling \texorpdfstring{$\boldsymbol{p}$}{p} and \texorpdfstring{$\boldsymbol{\varepsilon}$}{eps}}
Finally, the full conditional posteriors for $\boldsymbol{p}$ and $\varepsilon$ are conjugate to their Dirichlet and Beta priors:
\begin{align*}
    \boldsymbol{p}^0 \mid \mathbf{B}, \mathbf{W}, \mathbf{Z} &\sim \text{Dirichlet} \left(\frac{\alpha_{DP}}{L_0} + N_1^0  , \dots, \frac{\alpha_{DP}}{L_0} + N_{L_0}^0 \right) \\
    \boldsymbol{p}^{k,\ell} \mid \mathbf{B}, \mathbf{W}, \mathbf{Z} &\sim \text{Dirichlet} \left(\frac{\alpha_{DP}}{L} + N_1^{k,\ell}, \dots, \frac{\alpha_{DP}}{L} + N_L^{k,\ell} \right) \\
    \varepsilon \mid \mathbf{B}, \mathbf{W}, \mathbf{Z} &\sim \text{Beta} \left(1 + 
\sum_{h=1}^{L_0} N_h^0, 1 + \sum_{k=1}^K \sum_{\ell=1}^K \sum_{h=1}^{L} N_h^{k,\ell}\right).
\end{align*}

\section{Stochastic Variational inference}\label{ap:variationalinfdetails}

\subsection{Computing ELBO}

Under our modeling assumptions, the ELBO can be expressed according to the following formula:

\begin{equation}
\label{eq:elbo}
\begin{aligned}
    \operatorname{ELBO}_{\boldsymbol{\eta}} &= \sum_{\ell = 1}^K \operatorname{ELBO}_{\boldsymbol{\eta}} (\mu_\ell) + \sum_{k=1}^K \sum_{\ell=1}^K \operatorname{ELBO}_{\boldsymbol{\eta}}(\alpha_{k,\ell}) + \sum_{h=1}^{L_0} \mathrm{ELBO}_{\boldsymbol{\eta}}\left(p_h^0\right) + \mathrm{ELBO}_\eta\left(a_h^0\right) + \mathrm{ELBO}_\eta\left(b_h^0\right) \\
    &+ \sum_{k=1}^K \sum_{\ell=1}^K \sum_{h=1}^L  \operatorname{ELBO}_{\boldsymbol{\eta}}\left(p_h^{k, \ell}\right) + \operatorname{ELBO}_\eta\left(a_h^{k, \ell}\right) + \operatorname{ELBO}_\eta\left(b_h^{k, \ell}\right) + \mathrm{ELBO}_{\boldsymbol{\eta}}(\varepsilon) \\
    & + \sum_{k=1}^K \sum_{\ell=1}^K \sum_{(i, j) \in \mathcal{T}_{k, \ell}} \eta_{B_{i, j}}  \log \frac{T_b}{(t_i-t_j)(T_b - (t_i-t_j))} \\
    & + \sum_{h=1}^{L_0}\mathrm{E}_{q_{\eta_{a_h^0, b_h^0}}}\left[\log\frac{\Gamma\left(a_h^0+b_h^0\right)}{\Gamma\left(a_h^0\right) \Gamma\left(b_h^0\right)}\right] \mathrm{E}_{q_{\eta_{\mathrm{B}, \mathrm{w}, \mathrm{z}}}}\left[N_h^0\right] \\
    &+ \sum_{k=1}^K \sum_{\ell=1}^K\sum_{h=1}^{L}\left[\log\frac{\Gamma\left(a_h^{k, \ell}+b_h^{k, \ell}\right)}{\Gamma\left(a_h^{k, \ell}\right) \Gamma\left(b_h^{k, \ell}\right)}\right] \mathrm{E}_{q_{\eta_{\mathrm{B}, \mathrm{w}, \mathrm{z}}}}\left[N_h^{k, \ell}\right],
\end{aligned}
\end{equation}
%_{\text{non-conjugate for } \boldsymbol{a}, \boldsymbol{b}}
where
\begin{equation*}
\begin{aligned}
   \operatorname{ELBO}_{\boldsymbol{\eta}}(\mu_\ell) &= \left[\psi\left(\eta_{\mu, \ell, 1}\right)-\log \left(\eta_{\mu, \ell, 2}\right)\right]\left(\operatorname{E}_{q_{\eta_{\mathrm{B}}}}\left[\left|I_{ \ell}\right|\right]+a_\mu-\eta_{\mu, \ell, 1}\right) - \frac{\eta_{\mu, \ell, 1}}{\eta_{\mu, \ell, 2}}(T + b_\mu - \eta_{\mu,\ell, 2}) \\
  \operatorname{ELBO}_{\boldsymbol{\eta}}(\alpha_{k,\ell}) &= \left[\psi\left(\eta_{\alpha, k, \ell, 1}\right)-\log \left(\eta_{\alpha, k, \ell, 2}\right)\right]\left(\operatorname{E}_{q_{\eta_{\mathrm{B}}}}\left[\left|O_{k, \ell}\right|\right]+e_{k, \ell}-\eta_{\alpha, k, \ell, 1}\right)-\frac{\eta_{\alpha, k, \ell, 1}}{\eta_{\alpha, k, \ell, 2}}\left(n_k+f_{k, \ell}-\eta_{\alpha, k, \ell, 2}\right) \\
  \operatorname{ELBO}_{\boldsymbol{\eta}}(p_h^0) &= \left[\psi\left(\eta_{p, 0, h}\right)-\psi\left(\sum_{h=1}^{L_0} \eta_{p, 0, h}\right)\right]\left(\mathrm{E}_{q_{\eta_{\mathbf{B}, \mathbf{w}, \mathbf{z}}}}\left[N_h^0\right]+\frac{\alpha_{D P}}{L_0}-\eta_{p, 0, j}\right)\\
  \operatorname{ELBO}_{\boldsymbol{\eta}}(p_h^{k,\ell}) &= \left[\psi\left(\eta_{p, k,\ell, h}\right)-\psi\left(\sum_{h=1}^{L} \eta_{p, k, \ell, h}\right)\right]\left(\mathrm{E}_{q_{\eta_{\mathbf{B}, \mathbf{w}, \mathbf{z}}}}\left[N_h^{k,\ell}\right]+\frac{\alpha_{D P}}{L}-\eta_{p, k,\ell, j}\right) \\
  \operatorname{ELBO}_{\boldsymbol{\eta}}(\varepsilon) &= \left[\psi\left(\eta_{\varepsilon}\right)-\psi\left(1\right)\right]\left(\sum_{h=1}^{L_0}\mathrm{E}_{q_{\eta_{\mathbf{B}, \mathbf{w}, \mathbf{z}}}}\left[N_h^0\right]+1-\eta_{\varepsilon}\right) \\&+ \left[\psi\left(1-\eta_{\varepsilon}\right)-\psi\left(1\right)\right]\left(\sum_{k=1}^K \sum_{\ell=1}^K\sum_{h=1}^{L}\mathrm{E}_{q_{\eta_{\mathbf{B}, \mathbf{w}, \mathbf{z}}}}\left[N_h^{k, \ell}\right]+1-\eta_{\varepsilon}\right) \\
  \operatorname{ELBO}_{\boldsymbol{\eta}}(a_h^{k,\ell}) &= \frac{\eta_{a, k, \ell, h, 1}}{\eta_{a, k, \ell, h, 2}}\left[\sum_{(i, j) \in \mathcal{T}_{k, \ell}} \eta_{B_{i, j}}\eta_{W_{i, j}} \eta_{Z_{i,j,h}}\log \left(\frac{t_i -t_j}{T_b} \right) + d_a - \eta_{a, k, \ell, h, 2} \right] \\
  \operatorname{ELBO}_{\boldsymbol{\eta}}(b_h^{k,\ell}) &= \frac{\eta_{b, k, \ell, h, 1}}{\eta_{b, k, \ell, h, 2}} \left[\sum_{(i, j) \in \mathcal{T}_{k, \ell}} \eta_{B_{i, j}}\eta_{W_{i, j}}\eta_{Z_{i,j,h}}\log \left(1-\frac{t_i -t_j}{T_b} \right) + d_b -  \eta_{b, k, \ell, h, 2}\right] \\
    \operatorname{ELBO}_{\boldsymbol{\eta}}(a_h^{0}) &= \frac{\eta_{a, 0, h, 1}}{\eta_{a, 0, h, 2}}\left[\sum_{k=1}^K \sum_{\ell=1}^K\sum_{(i, j) \in \mathcal{T}_{k, \ell}} \eta_{B_{i, j}}\left(1- \eta_{W_{i, j}}\right) \eta_{Z_{i,j,h}}\log \left(\frac{t_i -t_j}{T_b} \right) + d_a - \eta_{a, 0, h, 2} \right] \\
  \operatorname{ELBO}_{\boldsymbol{\eta}}(b_h^{0}) &= \frac{\eta_{b, 0, h, 1}}{\eta_{b, 0, h, 2}} \left[\sum_{k=1}^K \sum_{\ell=1}^K\sum_{(i, j) \in \mathcal{T}_{k, \ell}} \eta_{B_{i, j}}\left(1- \eta_{W_{i, j}}\right) \eta_{Z_{i,j,h}}\log \left(1-\frac{t_i -t_j}{T_b} \right) + d_b -  \eta_{b, 0, h, 2}\right] .
\end{aligned}
\end{equation*}

%For simplicity, we first define and derive formula for the following auxilliary variables. %Suppose under the variational distribution at the $r$-th iteration, $a_{k, \ell, h} \sim \text{Gamma}(\eta_{a, k, \ell, h, 1}^{(r)}, \eta_{a, k, \ell, h, 2}^{(r)})$

\subsection{Updating \texorpdfstring{$\boldsymbol{\eta}_{\mu_p^0}, \boldsymbol{\eta}_{\mu_p^{k,\ell}}$}{etap}}
\begin{align*}
\eta_{p, 0, j}^{(r+1)} &= (1 - \rho_r)\eta_{p, 0, j}^{(r)} + \rho_r \left(\frac{\alpha_{DP}}{L_0} + \sum_{k=1}^K \sum_{\ell=1}^L \sum_{(i, j) \in \mathcal{T}_{k, i}} \eta^{(r)}_{B_{i, j}} \eta^{(r)}_{W_{i, j}} \eta^{(r)}_{Z_{i,j,h,0}} .\right) \\
\eta_{p, k, \ell, j}^{(r+1)} &= (1 - \rho_r)\eta_{p, k, \ell, j}^{(r)} + \rho_r \left(\frac{\alpha_{DP}}{L_0} + \sum_{(i, j) \in \mathcal{T}_{\mathcal{T}, \boldsymbol{i}}} \eta^{(r)}_{B_{i, j}}\left(1-\eta_{W_{i, j}}^{(r)}\right)\eta_{Z_{i,j,h,1}}^{(r)}.\right)   
\end{align*}

\subsection{Updating \texorpdfstring{$\boldsymbol{\eta}_{\varepsilon}$}{etavar}}
\begin{align*}
    \eta_{\varepsilon, 1}^{(r+1)} &= (1 - \rho_r)\eta_{\varepsilon, 1}^{(r)} + \rho_r \left(1 + \sum_{h=1}^{L_0} \sum_{k=1}^K \sum_{\ell=1}^L \sum_{(i, j) \in \mathcal{T}_{k, i}} \eta_{B_{i, j}}^{(r)} \eta_{W_{i, j}}^{(r)} \eta_{Z_{i, j, h, 0}}^{(r)}\right) \\
        \eta_{\varepsilon, 2}^{(r+1)} &= (1 - \rho_r)\eta_{\varepsilon, 2}^{(r)} + \rho_r \left(1 + \sum_{k=1}^K \sum_{\ell=1}^K \sum_{h=1}^L \sum_{(i, j) \in \mathcal{T}_{\mathcal{T}, i}} \eta_{B_{i, j}}^{(r)}\left(1-\eta_{W_{i, j}}^{(r)}\right) \eta_{Z_{i, j, h, 1}}^{(r)}\right)
\end{align*}

\subsection{Updating \texorpdfstring{$\boldsymbol{\eta}_{\mu_\ell}, \boldsymbol{\eta}_{\alpha_{k, \ell}}$}{etamu}}
\begin{align*}
    \eta_{\mu, \ell, 1}^{(r+1)} &= \left(1-\rho_r\right)\eta_{\mu, \ell, 1}^{(r)} + \rho_r \left(a_\mu + \sum_{i=1}^N\eta_{B_{i,i}}^{(r)}\right) \\
    \eta_{\mu, \ell, 2}^{(r+1)} &= \left(1-\rho_r\right)\eta_{\mu, \ell, 1}^{(r)} + \rho_r \left(b_\mu + T\right) \\
    \eta_{\alpha, k, \ell, 1}^{(r+1)} &= (1 - \rho_r) \eta_{\alpha, k, \ell, 1}^{(r)} + \rho_r \left( e_{k, \ell} + \sum_{d_i=k, d_j=\ell, i<j} \eta_{B_{i,j}}^{(r)}\right) \\
    \eta_{\alpha, k, \ell, 2}^{(r+1)} &= (1 - \rho_r) \eta_{\alpha, k, \ell, 2}^{(r)} + \rho_r \left(n_k + f_{k,\ell}\right)
\end{align*}
\subsection{Updating \texorpdfstring{$\boldsymbol{\eta}_{a_{k,\ell,h}}, \boldsymbol{\eta}_{b_{k,\ell,h}}$}{etamu}}
Applying the lower bound approximation using \eqref{eq:elbolb}, we obtain the conjugate structure for $a_h^{k, \ell}$:
\begin{align*}
   \mathrm{ELBO}_{\boldsymbol{\eta}}\left(a_h^{k, \ell}\right)&\geq \left[\psi\left(\eta_{a, k, \ell, h, 1}\right)-\log \left(\eta_{a, k, \ell, h, 2}\right)\right] \left\{ \left[\bar{a}_{k,\ell,h} \left(  \psi\left(\bar{a}_{k,\ell,h} + \bar{b}_{k,\ell,h}\right) - \psi\left(\bar{a}_{k,\ell,h}\right) \right) \right.\right.\\  
    &\left.\left. + \bar{a}_{k,\ell,h}\bar{b}_{k,\ell,h}\psi'\left(\bar{a}_{k,\ell,h} + \bar{b}_{k,\ell,h}\right) \left(\psi\left(\eta_{b, k, \ell, h, 1}\right)-\log \left(\eta_{b, k, \ell, h, 2}\right) - \log \bar{b}_{k,\ell,h}\right) \right]\right. \\
    &\left. \times\mathrm{E}_{q_{\eta_{\mathbf{B}, \mathbf{w}, \mathbf{z}}}}\left[N_h^{k, \ell}\right] + c_a - \eta_{a,k,\ell,h,1}   \right\} \\
    &+ \frac{\eta_{a, k, \ell, h, 1}}{\eta_{a, k, \ell, h, 2}}\left[\sum_{(i, j) \in \mathcal{T}_{k, \ell}} \eta_{B_{i, j}} \eta_{W_{i, j}} \eta_{Z_{i, y, h}} \log \left(\frac{t^{(r)}_i-t^{(r)}_j}{T_b}\right)+d_a-\eta_{a, k, \ell, h, 2}\right]
\end{align*}
\begin{align*}
   \mathrm{ELBO}_{\boldsymbol{\eta}}\left(b_h^{k, \ell}\right)&\geq \left[\psi\left(\eta_{b, k, \ell, h, 1}\right)-\log \left(\eta_{b, k, \ell, h, 2}\right)\right] \left\{ \left[\bar{b}_{k,\ell,h} \left(  \psi\left(\bar{a}_{k,\ell,h} + \bar{b}_{k,\ell,h}\right) - \psi\left(\bar{b}_{k,\ell,h}\right) \right) \right.\right.\\  
    &\left.\left. + \bar{a}_{k,\ell,h}\bar{b}_{k,\ell,h}\psi'\left(\bar{a}_{k,\ell,h} + \bar{b}_{k,\ell,h}\right) \left(\psi\left(\eta_{a, k, \ell, h, 1}\right)-\log \left(\eta_{a, k, \ell, h, 2}\right) - \log \bar{a}_{k,\ell,h}\right) \right]\right. \\
    &\left. \times\mathrm{E}_{q_{\eta_{\mathbf{B}, \mathbf{w}, \mathbf{z}}}}\left[N_h^{k, \ell}\right] + c_b - \eta_{b,k,\ell,h,1}   \right\} \\
    &+ \frac{\eta_{b, k, \ell, h, 1}}{\eta_{b, k, \ell, h, 2}}\left[\sum_{(i, j) \in \mathcal{T}_{k, \ell}} \eta_{B_{i, j}} \eta_{W_{i, j}} \eta_{Z_{i, y, h}} \log \left(1-\frac{t^{(r)}_i-t^{(r)}_j}{T_b}\right)+d_b-\eta_{b, k, \ell, h, 2}\right]
\end{align*}
We update $\eta_{a, k, \ell, h, 1}^{(r)}, \eta_{a, k, \ell, h, 2}^{(r)}$ by the following formula:
\begin{equation}
\begin{aligned}
    \eta_{a, k, \ell, h, 1}^{(r+1)} &= (1 - \rho_r)\eta_{a, k, \ell, h, 1}^{(r)} + \rho_r \left\{\left[\frac{\eta^{(r)}_{a, k, \ell, h, 1}}{\eta^{(r)}_{a, k, \ell, h, 2}} \left(  \psi\left(\frac{\eta^{(r)}_{a, k, \ell, h, 1}}{\eta^{(r)}_{a, k, \ell, h, 2}} + \frac{\eta^{(r)}_{b, k, \ell, h, 1}}{\eta^{(r)}_{b, k, \ell, h, 2}}\right) - \psi\left(\frac{\eta^{(r)}_{a, k, \ell, h, 1}}{\eta^{(r)}_{a, k, \ell, h, 2}}\right) \right) \right. \right.\\
    &\left.\left.  + \frac{\eta^{(r)}_{a, k, \ell, h, 1}}{\eta^{(r)}_{a, k, \ell, h, 2}}\frac{\eta^{(r)}_{b, k, \ell, h, 1}}{\eta^{(r)}_{b, k, \ell, h, 2}} \psi'\left(\frac{\eta^{(r)}_{a, k, \ell, h, 1}}{\eta^{(r)}_{a, k, \ell, h, 2}} +\frac{\eta^{(r)}_{b, k, \ell, h, 1}}{\eta^{(r)}_{b, k, \ell, h, 2}}\right)\left(\psi\left(\eta^{(r)}_{b, k, \ell, h, 1}\right)-\log \left(\eta^{(r)}_{b, k, \ell, h, 2}\right) - \log \frac{\eta^{(r)}_{b, k, \ell, h, 1}}{\eta^{(r)}_{b, k, \ell, h, 2}}\right)\right]\right.\\
    &\left.\times \kappa^{-1}\sum_{(i, j) \in \mathcal{T}_{k, \ell}} \eta^{(r)}_{B_{i, j}}\left(1-\eta^{(r)}_{W_{i, j}}\right) \eta^{(r)}_{Z_{i, j}} + c_a \right\}\\
    \eta_{a, k, \ell, h, 2}^{(r+1)} &= (1 - \rho_r)\eta_{a, k, \ell, h, 2}^{(r)} + \rho_r \left\{\kappa^{-1}\sum_{(i, j) \in \mathcal{T}_{k, \ell}} \eta_{B^{(r)}_{i, j}} \eta_{W^{(r)}_{i, j}} \eta_{Z^{(r)}_{i, j, h}} \log \left(\frac{t^{(r)}_i-t^{(r)}_j}{T_b}\right)+d_a \right\}.
\end{aligned}    
\end{equation}
Similarly, we can update $\eta_{b, k, \ell, h, 1}^{(r)}, \eta_{b, k, \ell, h, 2}^{(r)}$ as follows:
\begin{equation}
\begin{aligned}
    \eta_{b, k, \ell, h, 1}^{(r+1)} &= (1 - \rho_r)\eta_{b, k, \ell, h, 1}^{(r)} + \rho_r \left\{\left[\frac{\eta^{(r)}_{b, k, \ell, h, 1}}{\eta^{(r)}_{b, k, \ell, h, 2}} \left(  \psi\left(\frac{\eta^{(r)}_{a, k, \ell, h, 1}}{\eta^{(r)}_{a, k, \ell, h, 2}} + \frac{\eta^{(r)}_{b, k, \ell, h, 1}}{\eta^{(r)}_{b, k, \ell, h, 2}}\right) - \psi\left(\frac{\eta^{(r)}_{b, k, \ell, h, 1}}{\eta^{(r)}_{b, k, \ell, h, 2}}\right) \right) \right. \right.\\
    &\left.\left.  + \frac{\eta^{(r)}_{a, k, \ell, h, 1}}{\eta^{(r)}_{a, k, \ell, h, 2}}\frac{\eta^{(r)}_{b, k, \ell, h, 1}}{\eta^{(r)}_{b, k, \ell, h, 2}} \psi'\left(\frac{\eta^{(r)}_{a, k, \ell, h, 1}}{\eta^{(r)}_{a, k, \ell, h, 2}} +\frac{\eta^{(r)}_{b, k, \ell, h, 1}}{\eta^{(r)}_{b, k, \ell, h, 2}}\right)\left(\psi\left(\eta^{(r)}_{a, k, \ell, h, 1}\right)-\log \left(\eta^{(r)}_{a, k, \ell, h, 2}\right) - \log \frac{\eta^{(r)}_{a, k, \ell, h, 1}}{\eta^{(r)}_{a, k, \ell, h, 2}}\right)\right]\right.\\
    &\left.\times \kappa^{-1}\sum_{(i, j) \in \mathcal{T}_{k, \ell}} \eta^{(r)}_{B_{i, j}}\left(1-\eta^{(r)}_{W_{i, j}}\right) \eta^{(r)}_{Z_{i, j}} + c_b \right\}\\
    \eta_{b, k, \ell, h, 2}^{(r+1)} &= (1 - \rho_r)\eta_{b, k, \ell, h, 2}^{(r)} + \rho_r \left\{\kappa^{-1}\sum_{(i, j) \in \mathcal{T}_{k, \ell}} \eta_{B^{(r)}_{i, j}} \eta_{W^{(r)}_{i, j}} \eta_{Z^{(r)}_{i, j, h}} \log \left(1-\frac{t^{(r)}_i-t^{(r)}_j}{T_b}\right)+d_b \right\}.
\end{aligned}    
\end{equation}
We can update $\eta_{a, 0, h, 1}^{(r)}, \eta_{a, 0, h, 2}^{(r)}, \eta_{b, 0, h, 1}^{(r)}, \eta_{b, 0, h, 2}^{(r)}$ in the same way:
\begin{equation}
\begin{aligned}
    \eta_{a, 0, h, 1}^{(r+1)} &= (1 - \rho_r)\eta_{a,0,h, 1}^{(r)} + \rho_r \left\{\left[\frac{\eta^{(r)}_{a, 0, h, 1}}{\eta^{(r)}_{a,0, h, 2}} \left(  \psi\left(\frac{\eta^{(r)}_{a, k, \ell, h, 1}}{\eta^{(r)}_{a, 0, h, 2}} + \frac{\eta^{(r)}_{b, 0, h, 1}}{\eta^{(r)}_{b, 0, h, 2}}\right) - \psi\left(\frac{\eta^{(r)}_{a, 0, h, 1}}{\eta^{(r)}_{a, 0, h, 2}}\right) \right) \right. \right.\\
    &\left.\left.  + \frac{\eta^{(r)}_{a, 0, h, 1}}{\eta^{(r)}_{a, 0, h, 2}}\frac{\eta^{(r)}_{b, 0, h, 1}}{\eta^{(r)}_{b, 0, h, 2}} \psi'\left(\frac{\eta^{(r)}_{a, 0, h, 1}}{\eta^{(r)}_{a, 0, h, 2}} +\frac{\eta^{(r)}_{b, 0, h, 1}}{\eta^{(r)}_{b, 0, h, 2}}\right)\left(\psi\left(\eta^{(r)}_{b, 0, h, 1}\right)-\log \left(\eta^{(r)}_{b, 0, h, 2}\right) - \log \frac{\eta^{(r)}_{b, 0, h, 1}}{\eta^{(r)}_{b, 0, h, 2}}\right)\right]\right.\\
    &\left.\times \kappa^{-1} \sum_{k=1}^K \sum_{\ell=1}^K \sum_{(i, j) \in \mathcal{T}_{k, \ell}} \eta_{B_{i, j}^{(r)}} \eta_{W_{i, j}^{(r)}} \eta_{Z_{i, j, h}^{(r)}} + c_a \right\}\\
    \eta_{a, 0, h, 2}^{(r+1)} &= (1 - \rho_r)\eta_{a, 0, h, 2}^{(r)} + \rho_r \left\{\kappa^{-1}\sum_{k=1}^K \sum_{\ell=1}^K \sum_{(i, j) \in \mathcal{T}_{k, \ell}} \eta_{B^{(r)}_{i, j}}\left(1-\eta_{W^{(r)}_{i, j}}\right) \eta_{Z^{(r)}_{i, j, h}} \log \left(\frac{t^{(r)}_i-t^{(r)}_j}{T_b}\right)+d_a \right\}.
\end{aligned}    
\end{equation}
\begin{equation}
\begin{aligned}
    \eta_{b, 0, h, 1}^{(r+1)} &= (1 - \rho_r)\eta_{b,0,h, 1}^{(r)} + \rho_r \left\{\left[\frac{\eta^{(r)}_{b, 0, h, 1}}{\eta^{(r)}_{b,0, h, 2}} \left(  \psi\left(\frac{\eta^{(r)}_{a, k, \ell, h, 1}}{\eta^{(r)}_{a, 0, h, 2}} + \frac{\eta^{(r)}_{b, 0, h, 1}}{\eta^{(r)}_{b, 0, h, 2}}\right) - \psi\left(\frac{\eta^{(r)}_{b, 0, h, 1}}{\eta^{(r)}_{b, 0, h, 2}}\right) \right) \right. \right.\\
    &\left.\left.  + \frac{\eta^{(r)}_{a, 0, h, 1}}{\eta^{(r)}_{a, 0, h, 2}}\frac{\eta^{(r)}_{b, 0, h, 1}}{\eta^{(r)}_{b, 0, h, 2}} \psi'\left(\frac{\eta^{(r)}_{a, 0, h, 1}}{\eta^{(r)}_{a, 0, h, 2}} +\frac{\eta^{(r)}_{b, 0, h, 1}}{\eta^{(r)}_{b, 0, h, 2}}\right)\left(\psi\left(\eta^{(r)}_{a, 0, h, 1}\right)-\log \left(\eta^{(r)}_{a, 0, h, 2}\right) - \log \frac{\eta^{(r)}_{a, 0, h, 1}}{\eta^{(r)}_{a, 0, h, 2}}\right)\right]\right.\\
    &\left.\times \kappa^{-1} \sum_{k=1}^K \sum_{\ell=1}^K \sum_{(i, j) \in \mathcal{T}_{k, \ell}} \eta_{B_{i, j}^{(r)}} \eta_{W_{i, j}^{(r)}} \eta_{Z_{i, j, h}^{(r)}} + c_b \right\}\\
    \eta_{b, 0, h, 2}^{(r+1)} &= (1 - \rho_r)\eta_{b, 0, h, 2}^{(r)} + \rho_r \left\{\kappa^{-1}\sum_{k=1}^K \sum_{\ell=1}^K \sum_{(i, j) \in \mathcal{T}_{k, \ell}} \eta_{B^{(r)}_{i, j}}\left(1-\eta_{W^{(r)}_{i, j}}\right) \eta_{Z^{(r)}_{i, j, h}} \log \left(\frac{t^{(r)}_i-t^{(r)}_j}{T_b}\right)+d_b \right\}.
\end{aligned}    
\end{equation}

\subsection{Updating \texorpdfstring{$\boldsymbol{\eta}_{W}, \boldsymbol{\eta}_Z$}{etaW}}

%For each $(i,j) \in \mathcal{T}^{(r)}_{i,j}$, the 
%To derive the update formula for $\boldsymbol{\eta}_{W}$, we first derive an approximation for
%Suppose under the variational distribution, $a \sim \text{Gamma}(\eta_{a,1}, \eta_{a,2}), b \sim \text{Gamma}(\eta_{b,1}, \eta_{b,2})$, we let $Q(t \mid \eta_{a,1}, \eta_{a,2}, \eta_{b,1}, \eta_{b,2}):= \operatorname{E}_q \left[f_{\text{Beta}}\left(t \mid a,b; T_0 \right)\right]$, we have the following approximation formula:

Let $Q(t \mid \eta_{a,1}, \eta_{a,2}, \eta_{b,1}, \eta_{b,2}) := \operatorname{E}_q \left[f_{\text{Beta}}\left(t \mid a,b; T_0 \right)\right]$. Based on the approximation in Section \ref{sec:svi}, we have 

\begin{equation}
\begin{aligned}
\label{eq:Q}
    Q(t \mid &\eta_{a,1}, \eta_{a,2}, \eta_{b,1}, \eta_{b,2}) \approx -\log \operatorname{Be}\left(\frac{\eta_{a, 1}}{\eta_{a, 2}}, \frac{\eta_{b, 1}}{\eta_{b, 2}}\right) \Gamma\left(\frac{\eta_{a, 1}}{\eta_{a, 2}}\right) \Gamma\left(\frac{\eta_{b, 1}}{\eta_{b, 2}}\right)\\
    &+\frac{\eta_{a, 1}}{\eta_{a, 2}}\left[\psi\left(\frac{\eta_{a, 1}}{\eta_{a, 2}}+\frac{\eta_{b, 1}}{\eta_{b, 2}}\right)-\psi\left(\frac{\eta_{a, 1}}{\eta_{a, 2}}\right)\right]\left(\psi\left(\eta_{a, 1}\right)-\log \left(\eta_{a, 2}\right)-\log \frac{\eta_{a, 1}}{\eta_{a, 2}}\right)\\
    &+\frac{\eta_{b, 1}^1}{\eta_{b, 2}}\left[\psi\left(\frac{\eta_{a, 1}}{\eta_{a, 2}}+\frac{\eta_{b, 1}}{\eta_{b, 2}}\right)-\psi\left(\frac{\eta_{b, 1}}{\eta_{b, 2}}\right)\right]\left(\psi\left(\eta_{b, 1}\right)-\log \left(\eta_{b, 2}\right)-\log \frac{\eta_{b, 1}}{\eta_{b, 2}}\right)\\
    &+\frac{1}{2} \frac{\eta_{a, 1}}{\eta_{a, 2}}\left[\psi^{\prime}\left(\frac{\eta_{a, 1}}{\eta_{a, 2}}+\frac{\eta_{b, 1}}{\eta_{b, 2}}\right)-\psi^{\prime}\left(\frac{\eta_{a, 1}}{\eta_{a, 2}}\right)\right]\left[\left(\psi\left(\eta_{a, 1}\right)-\log \left(\eta_{a, 1}\right)\right)^2+\psi^{\prime}\left(\eta_{a, 1}\right)\right]\\
    &+\frac{1}{2}{\frac{\eta_{b, 1}}{\eta_{b, 2}}}^2\left[\psi^{\prime}\left(\frac{\eta_{a, 1}}{\eta_{a, 2}}+\frac{\eta_{b, 1}}{\eta_{b, 2}}\right)-\psi^{\prime}\left(\frac{\eta_{b, 1}}{\eta_{b, 2}}\right)\right]\left[\left(\psi\left(\eta_{b, 1}\right)-\log \left(\eta_{b, 1}\right)\right)^2+\psi^{\prime}\left(\eta_{b, 1}\right)\right]\\
    &+\frac{\eta_{a, 1}}{\eta_{a, 2}} \cdot \frac{\eta_{b, 1}}{\eta_{b, 2}} \cdot \psi^{\prime}\left(\frac{\eta_{a, 1}}{\eta_{a, 2}}+\frac{\eta_{b, 1}}{\eta_{b, 2}}\right)\left(\psi\left(\eta_{a, 1}\right)-\log \left(\eta_{a, 2}\right)-\log \frac{\eta_{a, 1}}{\eta_{a, 2}}\right)\left(\psi\left(\eta_{b, 1}\right)-\log \left(\eta_{b, 2}\right)-\log \frac{\eta_{b, 1}}{\eta_{b, 2}}\right) \\
    &+\left(\frac{\eta_{a, 1}}{\eta_{a, 2}}-1\right) \log t+\left(\frac{\eta_{b, 1}}{\eta_{b, 2}}-1\right) \log \left(T_0-t\right)-\left(\frac{\eta_{a, 1}}{\eta_{a, 2}}+\frac{\eta_{b, 1}}{\eta_{b, 2}}-1\right) \log T_0.
\end{aligned}
\end{equation}

\begin{align*}
    \eta_{Z_{i,j,h,0}}^{(r+1)} &\propto \frac{\psi(\eta^{(r+1)}_{p,0,h})}{ \psi(\sum_{h=1}^L\eta^{(r+1)}_{p,0,h})} + Q(t^{(r)}_j - t^{(r)}_i \mid \eta^{(r+1)}_{a_{0,h,1}}, \eta^{(r+1)}_{a_{0,h,2}}, \eta^{(r+1)}_{b_{0,h,1}}, \eta^{(r+1)}_{b_{0,h,2}}), ~~ h = 1, \dots, L_0. \\
    \eta_{Z_{i,j,h,1}}^{(r+1)} &\propto \frac{\psi(\eta^{(r+1)}_{p,k,\ell,h})}{ \psi(\sum_{h=1}^L\eta^{(r+1)}_{p,k,\ell,h})} + Q(t^{(r)}_j - t^{(r)}_i \mid \eta^{(r+1)}_{a_{k,\ell,h,1}}, \eta^{(r+1)}_{a_{k,\ell,h,2}}, \eta^{(r+1)}_{b_{k,\ell,h,1}}, \eta^{(r+1)}_{b_{k,\ell,h,2}}), ~~ h = 1, \dots, L. 
\end{align*}
Thus, we can update $\boldsymbol{\eta}_W$, $\boldsymbol{\eta}_Z$ using the following formula:
\begin{align*}
    \tilde{\eta}_{W_{i,j,1}}^{(r+1)} &= \sum_{h=1}^L \left[\frac{\psi(\eta_{\varepsilon, 2}^{(r+1)})}{\psi(\eta_{\varepsilon, 1}^{(r+1)} + \eta_{\varepsilon, 2}^{(r+1)})} + \frac{\psi(\eta^{(r+1)}_{p,k,\ell,h})}{ \psi(\sum_{h=1}^L\eta^{(r+1)}_{p,k,\ell,h})} + Q(t^{(r)}_j - t^{(r)}_i \mid \eta^{(r+1)}_{a_{k,\ell,h,1}}, \eta^{(r+1)}_{a_{k,\ell,h,2}}, \eta^{(r+1)}_{b_{k,\ell,h,1}}, \eta^{(r+1)}_{b_{k,\ell,h,2}})\right] \\
    \tilde{\eta}_{W_{i,j,0}}^{(r+1)} &= \sum_{h=1}^{L_0} \left[\frac{\psi(\eta_{\varepsilon, 1}^{(r+1)})}{\psi(\eta_{\varepsilon, 1}^{(r+1)} + \eta_{\varepsilon, 2}^{(r+1)})} + \frac{\psi(\eta^{(r+1)}_{p,0,h})}{ \psi(\sum_{h=1}^{L_0}\eta^{(r+1)}_{p,0,h})} + Q(t^{(r)}_j - t^{(r)}_i \mid \eta^{(r+1)}_{a_{0,h,1}}, \eta^{(r+1)}_{a_{0,h,2}}, \eta^{(r+1)}_{b_{0,h,1}}, \eta^{(r+1)}_{b_{0,h,2}})\right] \\
    \eta_{W_{i,j}}^{(r+1)} &= \frac{\tilde{\eta}_{W_{i,j,1}}^{(r+1)}}{\tilde{\eta}_{W_{i,j,0}}^{(r+1)} + \tilde{\eta}_{W_{i,j,1}}^{(r+1)}}
\end{align*}

\subsection{Updating \texorpdfstring{$\boldsymbol{\eta}_{B}$}{etaB}}

For $i = 1, \dots, n, ~~ j = 1, \dots, i$, we have:

\begin{align*}
    \tilde{\eta}^{(r+1)}_{B_{j,i}}&=\log \left(\frac{t_j^{(r)}-t_i^{(r)}}{T}\right)\left[\sum_{h=1}^L \eta_{Z_{i,j,h,1}}^{(r+1)}\eta_{W_{i,j}}^{(r+1)}\left(\frac{\eta_{a,k,\ell,h,1}^{(r+1)}}{\eta_{a,k,\ell,h,2}^{(r+1)}} -1 \right) + \sum_{h=1}^{L_0} \eta_{Z_{i,j,h,0}}^{(r+1)}(1-\eta_{W_{i,j}}^{(r+1)})\left(\frac{\eta_{a,0,h,1}^{(r+1)}}{\eta_{a,0,h,2}^{(r+1)}} -1 \right)\right] \\
    &+\log \left(1-\frac{t_j^{(r)}-t_i^{(r)}}{T}\right)\left[\sum_{h=1}^L \eta_{Z_{i,j,h,1}}^{(r+1)}\eta_{W_{i,j}}^{(r+1)}\left(\frac{\eta_{b,k,\ell,h,1}^{(r+1)}}{\eta_{b,k,\ell,h,2}^{(r+1)}} -1 \right) + \sum_{h=1}^{L_0} \eta_{Z_{i,j,h,0}}^{(r+1)}(1-\eta_{W_{i,j}}^{(r+1)})\left(\frac{\eta_{b,0,h,1}^{(r+1)}}{\eta_{b,0,h,2}^{(r+1)}} -1 \right)\right] \\&+\sum_{h=1}^{L_0} Q\left(t_j^{(r)}-t_i^{(r)} \mid \eta_{a_{0, h, 1}}^{(r+1)}, \eta_{a_{0, h, 2}}^{(r+1)}, \eta_{b_{0, h, 1}}^{(r+1)}, \eta_{b_{0, h, 2}}^{(r+1)}\right) (1-\eta_{W_{i,j}})\eta_{Z_{i,j,h,0}} \\
    &+ \sum_{h=1}^{L} Q\left(t_j^{(r)}-t_i^{(r)} \mid \eta_{a_{k, \ell, h, 1}}^{(r+1)}, \eta_{a_{k, \ell, h, 2}}^{(r+1)}, \eta_{b_{k, \ell, h, 1}}^{(r+1)}, \eta_{b_{k, \ell, h, 2}}^{(r+1)}\right)\eta_{W_{i,j}}\eta_{Z_{i,j,h,1}} \\
    &+ \psi(\eta_{\alpha,k,\ell,1}) -\log(\eta_{\alpha,k,\ell,2}) -\log T_0  \\
    \tilde{\eta}^{(r+1)}_{B_{j,j}}&= \psi(\eta_{\mu,k,1}) - \log(\eta_{\mu,k,2}).
\end{align*}
Finally, we have
\begin{align*}
    {\eta}^{(r+1)}_{B_{j,i}} &= \frac{\tilde{\eta}^{(r+1)}_{B_{j,i}}}{\sum_{j=1}^i\tilde{\eta}^{(r+1)}_{B_{j,i}}}.
\end{align*}

\section{Estimation results under the mis-specified scenario.}

\begin{table}[H]
\centering
\begin{tabular}{cccccccc}
\hline
 &
  \multicolumn{3}{c}{MCMC} &
  \multicolumn{3}{c}{SVI} &
  EM-BK \\ 
 $\varepsilon_{\text{true}}$
 &
  RANDOM &
  IDIO &
  COMMON &
  RANDOM &
  IDIO &
  COMMON &
  - \\ \hline
$0$ &
\begin{tabular}[c]{@{}c@{}}0.060\\ (0.010)\end{tabular} &
\begin{tabular}[c]{@{}c@{}}0.063\\ (0.009)\end{tabular} &
\begin{tabular}[c]{@{}c@{}}0.254\\ (0.001)\end{tabular} &
\begin{tabular}[c]{@{}c@{}}0.159\\ (0.029)\end{tabular} &
\begin{tabular}[c]{@{}c@{}}0.170\\ (0.018)\end{tabular} &
\begin{tabular}[c]{@{}c@{}}0.289\\ (0.010)\end{tabular} &
\begin{tabular}[c]{@{}c@{}}0.177\\ (0.018)\end{tabular} \\
$0.2$ &
\begin{tabular}[c]{@{}c@{}}0.059\\ (0.004)\end{tabular} &
\begin{tabular}[c]{@{}c@{}}0.062\\ (0.005)\end{tabular} &
\begin{tabular}[c]{@{}c@{}}0.204\\ (0.001)\end{tabular} &
\begin{tabular}[c]{@{}c@{}}0.152\\ (0.018)\end{tabular} &
\begin{tabular}[c]{@{}c@{}}0.164\\ (0.020)\end{tabular} &
\begin{tabular}[c]{@{}c@{}}0.221\\ (0.028)\end{tabular} &
\begin{tabular}[c]{@{}c@{}}0.165\\ (0.009)\end{tabular} \\
$0.5$ &
\begin{tabular}[c]{@{}c@{}}0.055\\ (0.013)\end{tabular} &
\begin{tabular}[c]{@{}c@{}}0.061\\ (0.007)\end{tabular} &
\begin{tabular}[c]{@{}c@{}}0.130\\ (0.002)\end{tabular} &
\begin{tabular}[c]{@{}c@{}}0.152\\ (0.017)\end{tabular} &
\begin{tabular}[c]{@{}c@{}}0.160\\ (0.022)\end{tabular} &
\begin{tabular}[c]{@{}c@{}}0.162\\ (0.038)\end{tabular} &
\begin{tabular}[c]{@{}c@{}}0.146\\ (0.004)\end{tabular} \\
$0.8$ &
\begin{tabular}[c]{@{}c@{}}0.044\\ (0.004)\end{tabular} &
\begin{tabular}[c]{@{}c@{}}0.063\\ (0.007)\end{tabular} &
\begin{tabular}[c]{@{}c@{}}0.058\\ (0.002)\end{tabular} &
\begin{tabular}[c]{@{}c@{}}0.165\\ (0.072)\end{tabular} &
\begin{tabular}[c]{@{}c@{}}0.128\\ (0.026)\end{tabular} &
\begin{tabular}[c]{@{}c@{}}0.103\\ (0.056)\end{tabular} &
\begin{tabular}[c]{@{}c@{}}0.150\\ (0.009)\end{tabular} \\
$1$ &
\begin{tabular}[c]{@{}c@{}}0.027\\ (0.004)\end{tabular} &
\begin{tabular}[c]{@{}c@{}}0.061\\ (0.009)\end{tabular} &
\begin{tabular}[c]{@{}c@{}}0.023\\ (0.004)\end{tabular} &
\begin{tabular}[c]{@{}c@{}}0.087\\ (0.060)\end{tabular} &
\begin{tabular}[c]{@{}c@{}}0.105\\ (0.012)\end{tabular} &
\begin{tabular}[c]{@{}c@{}}0.097\\ (0.064)\end{tabular} &
\begin{tabular}[c]{@{}c@{}}0.153\\ (0.007)\end{tabular} \\ \hline
\end{tabular}
\caption{RMISE as a point estimation accuracy metric for all methods under five true information-borrowing ratios for the mis-specified scenario. The values in the grid cells are the average over 10 independently generated datasets, and the standard deviation is shown in the brackets.}
\label{tab:pe-rmise2}
\end{table}

\begin{table}[H]
\centering
\begin{tabular}{ccccccc}
\hline
 &
  \multicolumn{3}{c}{MCMC} &
  \multicolumn{3}{c}{SVI} \\
 $\varepsilon_{\text{true}}$
 &
  RANDOM &
  IDIO &
  COMMON &
  RANDOM &
  IDIO &
  COMMON \\ \hline
$0$ &
\begin{tabular}[c]{@{}c@{}}0.781\\ (0.053)\end{tabular} &
\begin{tabular}[c]{@{}c@{}}0.792\\ (0.090)\end{tabular} &
\begin{tabular}[c]{@{}c@{}}0.044\\ (0.020)\end{tabular} &
\begin{tabular}[c]{@{}c@{}}0.091\\ (0.040)\end{tabular} &
\begin{tabular}[c]{@{}c@{}}0.108\\ (0.040)\end{tabular} &
\begin{tabular}[c]{@{}c@{}}0.064\\ (0.030)\end{tabular} \\ 
$0.2$ &
\begin{tabular}[c]{@{}c@{}}0.813\\ (0.071)\end{tabular} &
\begin{tabular}[c]{@{}c@{}}0.827\\ (0.078)\end{tabular} &
\begin{tabular}[c]{@{}c@{}}0.072\\ (0.047)\end{tabular} &
\begin{tabular}[c]{@{}c@{}}0.068\\ (0.014)\end{tabular} &
\begin{tabular}[c]{@{}c@{}}0.089\\ (0.045)\end{tabular} &
\begin{tabular}[c]{@{}c@{}}0.024\\ (0.006)\end{tabular} \\ 
$0.5$ &
\begin{tabular}[c]{@{}c@{}}0.773\\ (0.089)\end{tabular} &
\begin{tabular}[c]{@{}c@{}}0.840\\ (0.069)\end{tabular} &
\begin{tabular}[c]{@{}c@{}}0.200\\ (0.064)\end{tabular} &
\begin{tabular}[c]{@{}c@{}}0.047\\ (0.018)\end{tabular} &
\begin{tabular}[c]{@{}c@{}}0.080\\ (0.035)\end{tabular} &
\begin{tabular}[c]{@{}c@{}}0.027\\ (0.005)\end{tabular} \\ 
$0.8$ &
\begin{tabular}[c]{@{}c@{}}0.812\\ (0.06)\end{tabular} &
\begin{tabular}[c]{@{}c@{}}0.869\\ (0.072)\end{tabular} &
\begin{tabular}[c]{@{}c@{}}0.493\\ (0.057)\end{tabular} &
\begin{tabular}[c]{@{}c@{}}0.050\\ (0.021)\end{tabular} &
\begin{tabular}[c]{@{}c@{}}0.095\\ (0.020)\end{tabular} &
\begin{tabular}[c]{@{}c@{}}0.044\\ (0.024)\end{tabular} \\ 
$1$ &
\begin{tabular}[c]{@{}c@{}}0.840\\ (0.098)\end{tabular} &
\begin{tabular}[c]{@{}c@{}}0.860\\ (0.085)\end{tabular} &
\begin{tabular}[c]{@{}c@{}}0.759\\ (0.166)\end{tabular} &
\begin{tabular}[c]{@{}c@{}}0.052\\ (0.017)\end{tabular} &
\begin{tabular}[c]{@{}c@{}}0.119\\ (0.017)\end{tabular} &
\begin{tabular}[c]{@{}c@{}}0.052\\ (0.014)\end{tabular} \\ \hline
\end{tabular}
\caption{Coverage rate as an uncertainty estimation accuracy metric for all methods under five true information-borrowing ratios. The values in the grid cells are the average over 10 independently generated datasets, and the standard deviation is shown in the brackets.}
\label{tab:ue-acr2}
\end{table}

\begin{table}[H]
\centering
\begin{tabular}{ccccccc}
\hline
 &
  \multicolumn{3}{c}{MCMC} &
  \multicolumn{3}{c}{SVI} \\
 $\varepsilon_{\text{true}}$&
  RANDOM &
  IDIO &
  COMMON &
  RANDOM &
  IDIO &
  COMMON \\ \hline
$0$ &
\begin{tabular}[c]{@{}c@{}}0.003\\ (0.001)\end{tabular} &
\begin{tabular}[c]{@{}c@{}}0.003\\ (0.001)\end{tabular} &
\begin{tabular}[c]{@{}c@{}}0.057\\ (0.001)\end{tabular} &
\begin{tabular}[c]{@{}c@{}}0.046\\ (0.008)\end{tabular} &
\begin{tabular}[c]{@{}c@{}}0.047\\ (0.006)\end{tabular} &
\begin{tabular}[c]{@{}c@{}}0.071\\ (0.001)\end{tabular} \\ 
$0.2$ &
\begin{tabular}[c]{@{}c@{}}0.003\\ (<0.001)\end{tabular} &
\begin{tabular}[c]{@{}c@{}}0.003\\ (<0.001)\end{tabular} &
\begin{tabular}[c]{@{}c@{}}0.044\\ (0.001)\end{tabular} &
\begin{tabular}[c]{@{}c@{}}0.047\\ (0.006)\end{tabular} &
\begin{tabular}[c]{@{}c@{}}0.047\\ (0.008)\end{tabular} &
\begin{tabular}[c]{@{}c@{}}0.058\\ (0.009)\end{tabular} \\ 
$0.5$ &
\begin{tabular}[c]{@{}c@{}}0.004\\ (0.002)\end{tabular} &
\begin{tabular}[c]{@{}c@{}}0.003\\ (0.001)\end{tabular} &
\begin{tabular}[c]{@{}c@{}}0.026\\ (<0.001)\end{tabular} &
\begin{tabular}[c]{@{}c@{}}0.05\\ (0.008)\end{tabular} &
\begin{tabular}[c]{@{}c@{}}0.047\\ (0.008)\end{tabular} &
\begin{tabular}[c]{@{}c@{}}0.047\\ (0.016)\end{tabular} \\ 
$0.8$ &
\begin{tabular}[c]{@{}c@{}}0.003\\ (0.001)\end{tabular} &
\begin{tabular}[c]{@{}c@{}}0.003\\ (0.001)\end{tabular} &
\begin{tabular}[c]{@{}c@{}}0.008\\ (0.001)\end{tabular} &
\begin{tabular}[c]{@{}c@{}}0.053\\ (0.025)\end{tabular} &
\begin{tabular}[c]{@{}c@{}}0.038\\ (0.011)\end{tabular} &
\begin{tabular}[c]{@{}c@{}}0.034\\ (0.023)\end{tabular} \\ 
$1$ &
\begin{tabular}[c]{@{}c@{}}0.001\\ (<0.001)\end{tabular} &
\begin{tabular}[c]{@{}c@{}}0.003\\ (0.001)\end{tabular} &
\begin{tabular}[c]{@{}c@{}}0.001\\ (<0.001)\end{tabular} &
\begin{tabular}[c]{@{}c@{}}0.03\\ (0.025)\end{tabular} &
\begin{tabular}[c]{@{}c@{}}0.03\\ (0.003)\end{tabular} &
\begin{tabular}[c]{@{}c@{}}0.035\\ (0.026)\end{tabular} \\ \hline
\end{tabular}
\caption{Interval score as an uncertainty estimation accuracy metric for all methods under five true information-borrowing ratios. The values in the grid cells are the average over 10 independently generated datasets, and the standard deviation is shown in the brackets.}
\label{tab:ue-is2}
\end{table}

\section{Estimated Density Plots from the SVI algorithm}

\begin{figure}[H]
    \centering
    \includegraphics[width=\linewidth]{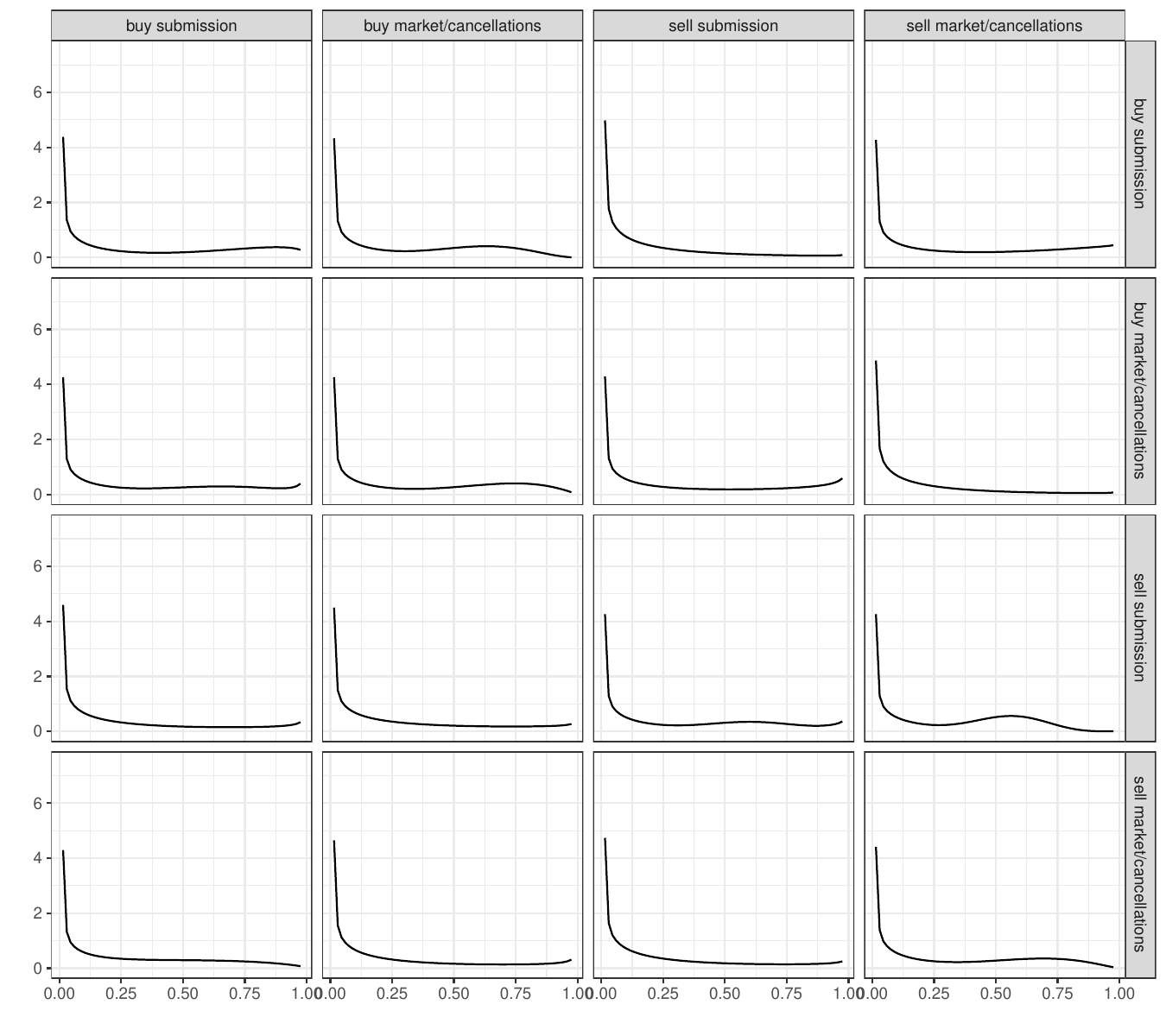}
    \caption{Estimated density plots from the SVI algorithm of the cross-dimensional inter-arrival times (within one second) among the four dimensions.}
    \label{fig:app-svi}
\end{figure}

%\begin{description}

%\item[Title:] Brief description. (file type)

%item[R-package for  MYNEW routine:] R-package ÒMYNEWÓ containing code to perform the diagnostic methods described in the article. The package also contains all datasets used as examples in the article. (GNU zipped tar file)

%\item[HIV data set:] Data set used in the illustration of MYNEW method in Section~ 3.2. (.txt file)

%\end{description}

%\section{BibTeX}

%We hope you've chosen to use BibTeX!\ If you have, please feel free to use the package natbib with any bibliography style you're comfortable with. The .bst file agsm has been included here for your convenience. 

\end{document}